\documentclass[prodmode,acmtissec]{acmsmall}  

\usepackage{amsfonts, amsmath, verbatim, algorithm, algpseudocode, latexsym}

\title{Probabilistic Analysis of Onion Routing in a Black-box Model}

\author{JOAN FEIGENBAUM \affil{Yale University} AARON JOHNSON
  \affil{U.S. Naval Research Laboratory} PAUL SYVERSON \affil{U.S. Naval
    Research Laboratory}
}

\begin{abstract}
We perform a probabilistic analysis of onion routing.  The analysis is presented in a black-box model of anonymous communication in the Universally Composable framework that abstracts the essential properties of onion routing in the presence of an active adversary that controls a portion of the network and knows all \emph{a priori} distributions on user choices of destination.  Our results quantify how much the adversary can gain in identifying users by exploiting knowledge of their probabilistic behavior.  In particular, we show that, in the limit as the network gets large, a user $u$'s anonymity is worst either when the other users always choose the destination $u$ is least likely to visit or when the other users always choose the destination $u$ chooses.  This worst-case anonymity with an adversary that controls a fraction $b$ of the routers is shown to be comparable to the best-case anonymity against an adversary that controls a fraction $\sqrt{b}$.
\end{abstract}



\begin{document}
\begin{bottomstuff}
Joan Feigenbaum (email: Joan.Feigenbaum@yale.edu) was supported in part by NSF grants 0331548 and 0534052, ARO grant W911NF-06-1-0316, and US-Israeli BSF grant 2002065.
Aaron Johnson (email: aaron.m.johnson@nrl.navy.mil) did the majority of this work while at Yale University and was supported by NSF grant 0428422 and ARO grant W911NF-05-1-0417. Some work was also done while at The University of Texas at Austin. 
Paul Syverson (email: syverson@itd.nrl.navy.mil) was supported by ONR.

An extended abstract of this paper appears in the Proceedings of the 2007 ACM Workshop on Privacy in the Electronic Society.
\end{bottomstuff}
\maketitle

\section{Introduction}
Every day, half a million people use the onion-routing network Tor~\cite{tor-design} to anonymize their Internet communication.  However, the effectiveness of this service, and of onion routing in general, is not well understood.  The approach we take to this problem is to model onion routing formally all the way from the protocol details to the behavior of the users.  We then analyze the resulting system and quantify the anonymity it provides.  Key features of our model include $i)$ a black-box abstraction in the Universally Composable (UC) framework~\cite{cryptoeprint:2000:067} that hides the underlying operation of the protocol and $ii)$ probabilistic user behavior and protocol operation.


Systems for communication anonymity generally have at most one of two desirable properties: provable security and practicality.  Systems that one can prove secure have used assumptions that make them impractical for most communication applications.  Practical systems are ultimately the ones we must care about, because they are the ones that will actually be used. However, their security properties have not been rigorously analyzed or even fully stated. This is no surprise, because practical anonymity systems have been deployed and available to study for perhaps a decade, while practical systems for communications confidentiality and/or authenticity have been in use almost as long as there have been electronic communications. It often takes a while for theory and practice to catch up to each other.

Of the many anonymous-communication design proposals (\emph{e.g.} \cite{chaum-mix,CHAUM,crowds:tissec,buses03,Salsa,Corrigan-Gibbs:2010:DAA:1866307.1866346}), onion routing
\cite{onion-routing:ih96} has had notable success in practice.
Several implementations have been made
\cite{onion-routing:ih96,onion-routing:pet2000,tor-design}, and there
was a similar commercial system, Freedom \cite{freedom1-arch}.  As of
October 2011, Tor \cite{tor-design}, the most recent iteration of the
basic design, consists of about 3000 routers, provides a total
bandwidth of over 1000 MB/s, and has an estimated total user
population of about 500,000~\cite{tor-metrics}.  Because of this
popularity, we believe it is important to improve our understanding of
the protocol.

Onion routing is a practical anonymity-network scheme with relatively low overhead and latency.  Users use a dedicated set of \emph{onion routers} to forward their traffic, obscuring the relationship between themselves and their destinations.  To communicate with a destination, a user selects a sequence of onion routers and constructs a \emph{circuit}, or persistant connection, over that sequence.  Messages to and from the destination are sent over the circuit.  Onion routing provides two-way, connection-based communication and does not require that the destination participate in the anonymity-network protocol.  These features make it useful for anonymizing much of the communication that takes place over the Internet today, such as web browsing, chatting, and remote login.  Thus, formal analysis and provable anonymity results for onion routing are significant.

As a step toward the overall goal of bridging the gap between provability and practicality in anonymous-communication systems, we have formally modeled and analyzed \emph{relationship anonymity} \cite{terminology,ShWa-Relationship} in Tor. Although this provides just a small part of the complete understanding of practical anonymity at which our research program is aimed, already it yields nontrivial results that require delicate probabilistic analysis. We hope that this aspect of the work will spur the Theoretical Computer Science community to devote the same level of attention to the rigorous study of anonymity as it has to the rigorous study of confidentiality. 

\subsection{Summary of Contributions}

{\bf Black-box abstraction:}
In the present paper, we treat the network simply as a ``black box''\footnote{We note that our use of a ``black box'' is slightly different than the more common uses in the literature. Black-box access to some cryptographic primitives is commonly used as a starting point to achieve some other desired functionality. Here we show how, for purposes of anonymity analysis, we need only consider a black-box abstraction.} to which users connect and
through which they communicate with destinations.  The abstraction
captures the relevant properties of a protocol execution that the
adversary can infer from his observations - namely, the observed
users, the observed destinations, and the possible connections between
the two.  In this way, we abstract away from much of the design
specific to onion routing so that our results apply both to onion
routing and to other low-latency anonymous-communication designs. 
We express the black-box model within the Universally Composable (UC) security
framework~\cite{cryptoeprint:2000:067}, which is a standard way to express the  
function and security properties of cryptographic protocols. We tie our functionality to 
the guarantees of an actual protocol by showing it reveals as much information about users' 
communication as the onion routing protocol we formalized~\cite{FC07} in an I/O-automata model. Moreover, we discuss how the functionality might be emulated by a protocol within the UC framework itself.

{\bf Probabilistic model:}
Our previous analysis in the I/O-automata model was possibilistic, a
notion of anonymity that is simply not sensitive enough. It makes no
distinction between communication that is equally likely to be from
any one of a hundred senders and communication that came from one
sender with probability $.99$ and from each of the other 99 senders
with probability $.000101$.  An adversary in the real world is likely
to have information about which scenarios are more realistic than
others.  In particular, users' communication patterns are not totally
random.  When the adversary can determine with high probability,
\emph{e.g.}, the sender of a message, that sender is not anonymous in
a meaningful way.

Using this intuition, we include a probability measure in our black-box model.  For any set of actual sources and destinations, there is a larger set that is consistent with the observations made by an adversary.  The adversary can then infer conditional probabilities on this larger set using the measure.  This gives the adversary probabilistic information about the facts we want the network to hide, such as the initiator of a communication.


In the probability measure that we use, each user chooses a destination according to some probability distribution.  We model heterogeneous user behavior by  allowing this distribution to be different for different users.  We also assume that the users choose their circuits by selecting the routers on it independently and at random.


After observing the protocol, the adversary
can in principle infer some distribution on circuit source and
destination.  He may not actually know the underlying probability
measure, however.  In particular, it doesn't seem likely that the
adversary would know how every user selects destinations.  In our
analysis, we take a worst-case view and assume that the adversary
knows the distributions exactly.  Also, over time he might learn a
good approximation of user behavior via the long-term intersection
attack \cite{DanSer04}.  In this case, it may seem as though anonymity
has been essentially lost anyway.  However, even when the adversary
knows how a user generally behaves, the anonymity network may make it
hard for him to determine who is responsible for any specific action,
and the anonymity of a specific action is what we are interested in.

{\bf Anonymity metric:} We analyze \emph{relationship anonymity} \cite{terminology,ShWa-Relationship} in our onion routing model.  Relationship anonymity is obtained when the adversary cannot identify the destination of a user.  In terms of the conventional subject/action specification for anonymity \cite{terminology}, we can take the action to be communication from a given user $u$ and the subject to be the destination.  Suggested probabilistic metrics for anonymity applied to this case include probability assigned to the correct destination \cite{crowds:tissec}, the entropy of the destination distribution \cite{Diaz02,Serj02}, and maximum probability within the destination distribution \cite{TOTH}, where the distribution in each case is a conditional distribution given the adversary's view.  We will use the probability assigned to the correct destination as our metric.  In part, this is because it is the simplest metric.  Also, any statements about entropy and maximum probability metrics only make loose guarantees about the probability assigned to the actual subject, a quantity that clearly seems important to the individual users.

We look at the value of this anonymity metric for a choice of
destination by a user.  Fixing a destination by just one user, say
$u$, does not determine what the adversary sees, however.  The
adversary's observations are also affected by the destinations chosen
by the other users and the circuits chosen by everybody.  Because
those variables are chosen probabilistically under the measure we
added, the anonymity metric will have its own distribution.  Several
statistics about this distribution might be interesting; in this
paper, we look at its expectation. Unlike other common anonymity
metrics, our approach lets a user judge how secure he can expect a
specific communication activity to be and thus whether to do it or
not.

{\bf Bounds on Anonymity:}

The distribution of the anonymity metric for a given user and destination depends on the other users' destination distributions.  If their distributions are very different from that of the given user, the adversary may have an easy time separating out the actions of the user.  If they are similar, the user may more effectively hide in the crowd. We provide the following results on a user's anonymity and its dependence on other user behavior:

\begin{enumerate}
\item We show that a standard approximation to our metric
provides a lower bound on it (Thm.~\ref{thm:lwrbnd}).
\item We show that the worst case for anonymity over other users' behavior is
when every other user either always visits the destinations the user
is otherwise least likely to visit or always visits his actual
destination (Cor.~\ref{cor:worst}). The former will be the worst case in most situations.
\item We give an asymptotic expression for our metric in the worst cases (Thm.~\ref{thm:de}). The limit of this expression in the 
most common worst case with an adversary controlling a fraction $b$ of the network is equal to the lower bound on the metric when the adversary controls a larger fraction $\sqrt{b}$ of the network. This is significantly worse than the standard analysis suggested, and shows the importance of carefully considering the adversary's knowledge of the system.
\item We consider anonymity in a more typical set of user distributions in which each user selects a destination from a common Zipfian distribution.  Because the users are identical, every user hides well among the others, and we show that, as the user population grows, the anonymity approaches the lower bound (Thm.~\ref{sec:typdist}). This shows you may be able to use the standard approximation with accurate results if you are able to make assumptions about user behavior.
\end{enumerate}

\subsection{Related Work}

Ours is not the first formalization of anonymous communication. Early formalizations used communicating sequential processes~\cite{schneider96}, graph theory and possible worlds~\cite{modular-approach}, and epistemic logic~\cite{GROUP,halpern-oneill-2003}.  These works focused primarily on formalizing the high-level concept of anonymity in communication. For this reason, they applied their formalisms to toy examples or systems that are of limited practical application and can only provide very strong forms of anonymity, \emph{e.g.}, dining-cryptographers networks.  Also, with the exception of \citeN{halpern-oneill-2003}, they have at most a limited ability to represent probability and probabilistic reasoning. We have focused in \citeN{FC07} on formalizing a widely deployed and used, practical, low-latency system.

Halpern and O'Neill~\shortcite{halpern-oneill-2003} give a general formulation of anonymity in systems that applies to our model.  They describe a ``runs-and-systems'' framework that provides semantics for logical statements about systems.  They then give several logical definitions for varieties of anonymity.  It is straightforward to	apply this framework to the network model and protocol that we give in~\cite{FC07}.  Our possibilistic definitions of sender anonymity, receiver anonymity, and relationship anonymity then correspond to the notion of ``minimal anonymity'' as defined in their paper.  The other notions of anonymity they give are generally too strong and are not achieved in our model of onion routing.

Later formalizations of substantial anonymous communication
systems~\cite{camlys05,MAUW,Wikstrom04} have not been directly based on the
design of deployed systems and have focused on provability without
specific regard for applicability to an implemented or implementable
design. Also, results in these papers are for message-based
systems: each message is constructed to be processed as a
self-contained unit by the appropriate router, typically using the
generally available public encryption key for that router. Such
systems typically employ mixing, changing the appearance and
decoupling the ordering of input to output messages at the router to
produce anonymity locally~\cite{chaum-mix}.  Onion routing, on the
other hand, is circuit based: before passing any messages with user
content, onion routing first lays a circuit through the routers that
provides those routers the keys to be used in processing the actual
messages. Mixing can be combined with onion routing in various
ways~\cite{onion-routing:jsac98}, although this is not
typical~\cite{tor-design}. Such circuit creation facilitates
bidirectional, low-latency coommunication and has been an identifying
feature of onion routing since the first public use of the
phrase~\cite{onion-routing:ih96}. Thus, while illuminating and
important works on anonymous communication, the formalizations
above are not likely to be applicable to low-latency communications, and,
despite the title of \cite{camlys05}, are not analyses of onion
routing.

In this paper, we add probabilistic analysis to the framework of \citeN{FC07}. Other works have presented probabilistic analysis of anonymous communication~\cite{crowds:tissec,SHMAT,WRIGHT,statistical-disclosure,DanSer04,e2e-traffic,stop-and-go} and even of onion routing~\cite{onion-routing:pet2000}.  The work of Shmatikov and Wang~\shortcite{ShWa-Relationship} is particularly similar to ours.  It calculates relationship anonymity in mix networks and incorporates user distributions for selecting destinations.  However, with the exception of~\cite{SHMAT}, these have not been formal analyses. Also, whether for high-latency systems such as mix networks, or low-latency systems, such as Crowds and onion routing, many of the attacks in these papers are some form of intersection attack. In an intersection attack, one watches repeated communication events for patterns of senders and receivers over time. Unless all senders are on and sending all the time (in a way not selectively blockable by an adversary) and/or all receivers receiving all the time, if different senders have different receiving partners, there will be patterns that arise and eventually differentiate the communication partners. It has long been recognized that no system design is secure against a long-term intersection attack. Several of these papers set out frameworks for making that more precise. In particular, \cite{statistical-disclosure}, \cite{DanSer04}, and \cite{e2e-traffic} constitute a progression towards quantifying how long it takes (in practice) to reveal traffic patterns in realistic settings.

We are not concerned herein with intersection attacks. We are effectively assuming that the intersection attack is done. The adversary already has a correct distribution of a user's communication partners.  We are investigating the anonymity of a communication in which a user communicates with one of those partners in the distribution.  This follows the anonymity analyses performed in much of the literature \cite{stop-and-go,MAUW,crowds:tissec,onion-routing:pet2000}, which focus on finding the source and destination of an individual communication.  Our analysis differs in that we take into account the probabilistic nature of the users' behavior.

We expect this to have potential practical applications. For example, designs for shared security-alert repositories to facilitate both forensic analysis for improved security design and quicker responses to widescale attacks have been proposed~\cite{LINCOLN04}.  A participant in a shared security-alert repository might expect to be known to communicate with it on a regular basis. Assuming reports of intrusions, etc., are adequately sanitized, the concern of the participant should be to hide when it is that updates from that participant arrive at the repository, \emph{i.e.}, which updates are likely to be from that participant as opposed to others.

\section{Technical Preliminaries}
\subsection{Model}
We describe our analysis of onion routing in terms of an ideal functionality in the Universal Composability framework \cite{cryptoeprint:2000:067} 
We use such a functionality for
three reasons: First, it abstracts away the details that aren't relevant to anonymity,
second, it precisely expresses the cryptographic protocol properties that are necessary for our analysis to apply, and
third, it immediately suggests ways to perform similar
analyses of other anonymous-communication protocols that may not strictly provide this functionality.

In the onion routing protocol on which we base our model, users choose
from a generally known set of onion routers a subset that will
comprise a circuit for communicating anonymously.  Circuit
construction has been done in various ways throughout the history of
onion routing.  In the first version of onion
routing~\cite{onion-routing:ih96}, and other early
versions~\cite{onion-routing:jsac98,freedom1-arch}, after a user selects a
sequence of onion routers from a publicly-known set, the user then creates a
circuit through this sequence using an \emph{onion}, a data structure
effectively composed only of layers with nothing in the middle. There
is one public-key-encrypted layer for each hop in the circuit, the
decryption of which contains the identity of the next hop in the
circuit (if there is one) and keying material for passing data over
the established circuit. In later protocols, such as used in
Cebolla~\cite{cebolla} and Tor~\cite{tor-design}, the circuit is built
via a telescoping protocol that extends the circuit hop-by-hop, using
the existing circuit for each extension. For all of these, each hop
only communicates with the routers before and after it in the
sequence, and the messages are encrypted once for each router in the
circuit so that no additional information leaks about the identities
of the other routers or the destination of the circuit.
Cryptographic techniques are used so that message forgery is
countered. Some later designs returned to the non-interactive
circuit construction of the
original~\cite{overlier:pet2007,pairing:pet2007}. It is trivial to see
that all of these fit directly within our model.

Some versions of onion routing, such as those that do iterative
discovery of onion routers via a
DHT~\cite{tarzan:ccs02,ccs09-shadowwalker,ccs09-torsk}, will not fit
within our model without some extensions that we do not pursue herein.
This is because the probability of first-last router choice and router
compromise within a circuit can no longer be assumed to be
independent. Some anonymity protocols that do not use onion routing
may nonetheless also fit within our model, appropriately extended.
For example, in Crowds~\cite{crowds:tissec}, the adversary can learn
from observing the first and last routers, but the connection to the
first router does not automatically identify the source. On the other
hand the destination is always know to every router in the
circuit. The probability that an observed circuit predecessor is the
source can thus be combined with the observed destination and the a
priori source-destination probability distribution.

The adversary is computationally bounded, non-adaptively compromises
an unknown subset of the onion routers, and can actively attack the
protocol. The design of our functionality is based on the assumption that the
ways that the adversary can narrow down the possible mappings of users
to destinations is determined by the set of circuits for which he
controls the first router and the set of circuits for which he
controls the last router. This assumption comes from the results of \citeN{FC07},
which we explicitly relate to our ideal functionality in Sec.~\ref{subsec:fc2uc}.


Our ideal functionality models anonymous communication over some period of time. It takes as input from each user the identity of a destination. For every such connection between a user and destination, the functionality may reveal to the adversary identity of the user, the identity of the destination, or both. Revealing the user corresponds in onion routing to the first router in the circuit being compromised, and revealing the destination corresponds to the last router being compromised. We note that we include only information flow to the adversary in this functionality rather than try to capture the type of communication primitive offered by onion routing because our focus is analyzing anonymity rather than defining a useful anonymous-communication functionality. This model is reminiscent of the general model of anonymous communication used by Kesdogan et al.~\shortcite{limits-open} in their analysis of an intersection attack. However, we do make a few assumptions that are particularly appropriate for onion routing.

First, the functionality allows the adversary to know whether or not he has directly observed the user. This is valid under the assumption that the initiating client is not located at an onion router itself. This is the case for the vast majority of circuits in Tor and in all significant deployments of onion routing and similar systems to date. We discuss this assumption further in Section~\ref{conclusions}.

Second, we assume that every user is
responsible for exactly one connection in a round.  Certainly users can
communicate with multiple destinations simultaneously in actual
onion-routing systems.  However, it seems likely that in practice most
users have at most some small (and fixed-bound) number of active
connections at any time. To the extent that multiple connections are
only slightly more likely to be from the same user than if all
connections were independently made and identically distributed, this
is a reasonable approximation. This is increasingly true as the
overall number of connections grows. To the extent that multiple
connections are less likely to be from the same user this is a
conservative assumption that gives the adversary as much power to
break anonymity as the limited number of user circuits can provide.

Third, the functionality omits the possibility that the adversary observes the user and destination but does not recognize that 
those observations are part of the some connection. This is another
conservative assumption that is motivated by the existence of timing
attacks that an active adversary can use to link traffic that it sees
at various points along its path through the network
\cite{onion-routing:pet2000}.  In a timing attack, the adversary
observes the timing of the messages going into the onion-routing
network and matches them to similar patterns of messages coming out of
the onion-routing networks slightly later.  Such attacks have been
experimentally demonstrated \cite{hs-attack06,bauer:wpes2007} and are
easy to mount.

Note that our model does not capture several known attacks on
anonymity in onion routing.  In particular, it does not include
attacks exploiting resource interference \cite{torta05,ccs06-hotclockskew},
heterogeneity on network latency \cite{tissec-latency-leak},
correlated destinations between rounds, and identifying patterns of
communication \cite{ccsw09-fingerprinting}.  We do not include such
attacks primarily to focus on the most important threats to anonymity,
because many of the omitted attacks are attacks on underlying systems
rather than on the protocol (e.g., interference) or have limited
effectiveness or are mitigated by improvements to the protocol.  Also,
we see the analysis of our simplified model as a first step in
establishing rigorous guarantees of anonymity in increasingly
realistic models.


Let $U$ be the set of users with $|U|=n$.  Let $\Delta$ be the set of destinations. Let $R$ be the set of onion routers. Let $\mathcal{F}_{OR}$ be the ideal functionality. $\mathcal{F}_{OR}$ takes the set $A$ of compromised parties from the adversary at the beginning of the execution. Let $b = |A\cap R|/|R|$. When user $u$ forwards his input from the environment to $\mathcal{F}_{OR}$, the functionality checks to see if it is some $d\in \Delta$. If so, $\mathcal{F}_{OR}$ sends to the adversary one of the following, choosing each with the probability given:

\begin{tabular}{ll}
\centering
(1) & $(\bot,\bot)$ with probability $(1-b)^2$\\
(2) & $(u, \bot)$ with probability $b (1-b)$\\
(3) & $(\bot, d)$ with probability $(1-b) b$\\
(4) & $(u,d)$ with probability $b^2$.
\end{tabular}



To analyze the anonymity provided by the ideal functionality, we make two assumptions about the inputs from the environment. First, we assume that the environment selects the destination of user $u$ from a distribution $p^u$ over $\Delta$, where we denote the probability that $u$ chooses $d$ as $p^u_d$. Second, we assume that the environment sends a destination to each user. Note that these assumptions need not be made when showing that a protocol UC-emulates $\mathcal{F}_{OR}$. 

We refer to the combination of the adversary model, the assumptions about the environment, and the ideal functionality as the \emph{black-box model}. Let $C$ be the relevant \emph{configuration} resulting from an execution. $C$ includes a selection of a destination by each user, $C_D : U\rightarrow \Delta$, a set of users whose inputs are observed, $C_I : U\rightarrow \{0,1\}$, and a set of users whose outputs are observed, $C_O : U\rightarrow \{0,1\}$. A user's input, output, and destination will be called its \emph{circuit}.

For any configuration, there is a larger set of configurations that are consistent with the outputs that the adversary receives from $\mathcal{F}_{OR}$.  We will call two configurations \emph{indistinguishable} if the sets of inputs, outputs, and links between them that the adversary receives are the same.
We use the notation $C\approx \overline{C}$ to indicate that configurations $C$ and $\overline{C}$ are indistinguishable.

\subsection{Probabilistic Anonymity}
A user performs an action anonymously in a possibilistic sense if there is an indistinguishable configuration in which the user does not perform the action.  For example, under this definition a user with observed output but unobserved input sends that output anonymously if there exists another user with unobserved input.  The probability measure we have added to configurations allows us to incorporate the degree of certainty that the adversary has about the subject of an action.  After making observations in the actual configuration, the adversary can infer a conditional probability distribution on configurations.  There are several candidates in the literature for assessing an anonymity metric from this distribution.  The probabilistic anonymity metric that we use is the posterior probability of the correct subject.  The lower this is, the more anonymous we consider the user.

\subsection{Relationship Anonymity}
We analyze the \emph{relationship anonymity} of users and destinations in our model, that is, how well the adversary can determine if a user and destination have communicated.  Our metric for the relationship anonymi\-ty of user $u$ and destination $d$ is the posterior probability $\psi$ that $u$ chooses $d$ as his destination.  We study $\psi$ directly, although the \emph{anonymity} of a user's communication with a destination is $1-\psi$.

Using the posterior probability makes sense in this context because it
focuses on the information that users are trying to hide---their
actual destinations---without being affected by information the
adversary learns about other destinations.  Onion routing does leak
information, and using a metric such as the entropy of the posterior
distribution or the statistical distance from the prior may not give a
good idea of how well the adversary's can correctly guess the user's
behavior.  Designers may wish to know how well a system protects
communications on average or overall. But it is also important for a
user to be able to assess how secure he can expect a particular
communication to be in order to decide whether to create it or
not. This is the question we address. Moreover, the metric is
relatively simple to analyze. Furthermore, to the extent that the user
may not know how he fits in and thus wishes to know the worst risk for
any user, that is just a lower bound on our metric.

The relationship anonymity of $u$ and $d$ varies with the destination
choices of the other users and the observations of the adversary.  If,
for example, $u$'s output is observed, and the inputs of all other
users are observed, then the adversary knows $u$'s destination with
probability 1.  Because we want to examine the relationship anonymity
of $u$ conditioned only on his destination, we end up with a
distribution on the anonymity metric.  We look at the expectation of
this distribution.  Moreover, because this distribution depends on the
destination distributions of all of the users, we continue by finding
the worst-case expectation in the limit for a given user and
destination and then examine the expectation in a more likely
situation.


\subsection{Emulating the Ideal Functionality} \label{subsec:fc2uc}
The anonymity analysis of the ideal functionality $\mathcal{F}_{OR}$ that we perform in Sections~\ref{sec:expanon} and \ref{sec:typdist} is meaningful to the extent that $\mathcal{F}_{OR}$ captures the information that an adversary can obtain by interacting with onion-routing protocols. We justify the functionality primarily by showing that it provides the same information about the source of a given connection as onion-routing as formalized by \citeN{FC07}. Furthermore, towards a more standard cryptographic analysis, we describe the way in which it should be possible to UC-emulate $\mathcal{F}_{OR}$, although we do not provide such a result here.

{\bf Relationship to I/O-automata model}
\citeN{FC07} formalize onion routing using an I/O-automata model\cite{LYNCH} and an idealization of the cryptographic properties of the protocol. Their analysis identifies the user states that are information-theoretically indistinguishable. The black-box model we provide herein is a valid abstraction of that formalization because, under some reasonable probability measures on executions, it preserves the relationship-anonymity properties.

The I/O-automata model includes a set of users $U$, a set of routers $R$, an adversary $A\subseteq R$, and a set of destinations $\Delta$, where we take the final router in the I/O-automata model to be the destination and assume that it is uncompromised.  A configuration $C$ in the I/O-automata model is a mapping from each user $u\in U$ to a circuit $(r^u_1,\ldots ,r^u_l)\in R^l$, a destination $d^u\in \Delta$, and a circuit identifier $n^u\in \mathbb{N}_+$.  An \emph{execution} is a sequence of I/O-automaton states and actions, which must be consistent with the configuration.

Let users in the I/O-automata model choose the other routers in their circuits uniformly at random and choose the destination according to user-specific distributions.  Given these circuits and a set of adversary automata, \citeN{FC07} identifies an equivalence class of circuit and destination choices such that, for every pair of configurations in the class, a bijection exists between their executions such that paired executions are indistinguishable.  Let the indistinguishable executions thus paired have the same probability, conditional on their configuration.

Given this measure, the black-box model that abstracts the I/O-automata model has the same user set $U$, the same destination set $\Delta$, an adversary parameter of $b = |A|/|R|$, and the same destination distributions.  The following theorem shows that each posterior distribution on the destinations of users has the same probability under both the I/O-automata model and its black-box model.  Let $E$ be a random I/O-automata execution.  Let $X^a$ be a random I/O-automata configuration ($X^a$ can be viewed as a function mapping a random execution to its configuration).  Let $X^b$ be a random black-box configuration.  Let $\psi_1(u,d,E)$ be the posterior probability that $u$ visited $d$ in the I/O-automata model, \emph{i.e.}, the conditional given that the execution is indistinguishable from $E$.  Let $\psi_2(u,d,X^b)$ be the posterior probability that $u$ visited $d$ in the black-box model, \emph{i.e.}, the conditional distribution given that the configuration is indistinguishable from $X^b$.  Let $\psi_0(u,d)$ be a distribution over destinations $d$ for every $u$.
\begin{theorem}
\begin{equation*}
Pr[\forall_{u\in U, d\in \Delta} \psi_1(u,d,E) = \psi_0(u,d)] = Pr[\forall_{u\in U, d\in \Delta} \psi_2(u,d,X^b) = \psi_0(u,d)]
\end{equation*}
\end{theorem}
\begin{proof}
Let $\phi$ be the map from I/O-automata configurations to black-box configurations such that
\begin{enumerate}
\item $\phi(C)_D(u) = d^u$
\item $\phi(C)_I(u) = \left\{ \begin{array}{ll} 1 & \textrm{ if } r_1\in A\\ 0 & \textrm{ otherwise } \end{array} \right.$
\item $\phi(C)_O(u) = \left\{ \begin{array}{ll} 1 & \textrm{ if } r_l\in A\\ 0 & \textrm{ otherwise } \end{array} \right.$.
\end{enumerate}
$\phi$ essentially ``quotients out''  the specific router choices of each user, retaining the compromised status of the first and last routers as well as the destination.   It allows us to relate the posterior $\psi_1$ in the I/O-automata model to the $\psi_2$ in the black-box model.

Let $C_1^a$ be any I/O-automata configuration.  Given any execution $e$ of $C_1^a$, the adversary's posterior probability on configurations is
\begin{equation*}
\frac{Pr[X^a = C_2^a]}{\sum_{C^a \approx C_1^a} Pr[X^a = C^a]}
\end{equation*}
if $C_2^a \approx C_1^a$ and $0$ otherwise, because we set equal the probability of two executions that are paired with each other in the bijection on executions constructed in \citeN{FC07}.  Because the configurations determine which destination each user visits, the distribution $\psi_1(u,d,e)$ can be determined from the posterior distribution on configurations.  Notice that this distribution only puts positive probability on the set $\mathcal{C}^a$ of configurations that are indistinguishable from $C_1^a$.

The posterior distribution on I/O-automata configurations induces a posterior distribution on black-box configurations via $\phi$.  $\phi$ preserves the destination of each user, and so the distribution $\psi_1(u,d,e)$ can be determined from this distribution on black-box configurations.  Notice that this distribution only puts positive probability on the set of black-box configurations $\phi(\mathcal{C}^a)$ that are mapped to by I/O-automata configurations in $\mathcal{C}^a$.

To understand the set $\phi(\mathcal{C}^a)$ and its posterior distribution given $e$, consider the equivalence class $\mathcal{C}^b$ of the configuration $\phi(C_1^a)$.  Let $S$ be those configurations in $\mathcal{C}^a$ that differ from $C_1^a$ only in the destinations and the permutation of users.   From Theorems 1 and 2 in \citeN{FC07}, it follows that $\phi$ is a bijection between $S$ and $\mathcal{C}^b$.  The posterior probability of each $C_2^a\in S$ is proportional to $Pr[X^b = \phi(C_2^a)]$ because the prior probability of $C_2^a$ is $Pr[X^b = \phi(C_2^a)]$ multiplied by the probability selecting its given routers (which are the same for all $s\in S$) given that $\phi(X^a)=\phi(C_2^a)$.  Moreover, all of the other configurations in $C^a$ are reached by changing the unobserved routers of one of the configurations in $S$.  $\phi$ is invariant under such a change.  Also, the posterior probability is invariant under such a change because the routers are chosen independently and uniformly at random.  Furthermore, the number of I/O-automata configurations that are reached by such a change from some $s\in S$ is the same for all $s$.  Therefore, the posterior probability $Pr[\phi(X^a)=C^b | e]$ is proportional to $Pr[X^b=C^b]$ for $C^b\in \mathcal{C}^b$, and is zero otherwise.  Therefore, $\psi_1(u,d,e) = \psi_2(u,d,\phi(C_1^a))$.

By this equality, the probability that a random execution $E$ results in a given posterior $\psi_0(u,d)$ is equal to the probability that the I/O-automata configuration $X^a$ maps under $\phi$ to a black-box configuration $\phi(X^a) = C^b$ such that $\psi_2(u,d,C^b)=\psi_0(u,d)$.  The probability $Pr[\phi(X^a) = C^b]$ is equal to $Pr[X^b = C^b]$ because the probability of first-router compromise and the probability of an input being observed are both $b$, last-router compromise and an output being observed are both independent events with probability $b$, and user destinations are chosen independently in both models and follow the same distributions.  Therefore,
\begin{equation*}
Pr[\forall_{u\in U, d\in \Delta} \psi_1(u,d,E) = \psi_0(u,d)] = Pr[\forall_{u\in U, d\in \Delta} \psi_2(u,d,X) = \psi_0(u,d)].
\end{equation*}
\end{proof}

{\bf UC-emulation}
Expressing our black-box model within the UC framework allows it to be compared to protocols expressed within the same framework. In particular, if a protocol can be shown to UC-emulate $\mathcal{F}_{OR}$, then, making only common cryptographic assumptions, the adversary can make only negligibly better guesses about users' communication when interacting with that protocol than he can with the functionality. The results of \citeN{camlys05} suggest that such emulation is indeed possible. An onion routing protocol similar to their protocol combined with a message transmission functionality that hides messages not to corrupt parties (cf. \citeN{cryptoeprint:2000:067}), should indeed hide the routers that are not corrupt or next to corrupt routers on a circuit. Then $\mathcal{F}_{OR}$ provides the adversary with all the information about user inputs that a simulator needs in order to simulate the rest of the protocol and confuse the adversary.

\section{Expected Anonymity} \label{sec:expanon}
Let the set $\mathcal{C}$ of all configurations be the sample space and $X$ be a random configuration. Let $\Psi$ be the posterior probability of the event that $u$ chooses $d$ as a destination, that is, $\Psi(C) = Pr[X_D(u)=d | X\approx C]$.  $\Psi$ is our metric for the relationship anonymity of $u$ and $d$.

Let $\mathbb{N}^{\Delta}$ represent the set of multisets over $\Delta$.  Let $\rho(\Delta^0)$ be the maximum number of orderings of $\Delta^0 \in \mathbb{N}^{\Delta}$ such that the same destination is in any given location in every ordering:
\begin{equation*}
\rho(\Delta^0) = \prod_{\delta \in \Delta} |\{\delta \in \Delta^0\}|!
\end{equation*}

Let $\Pi(A,B)$ be the set of all injective maps $A\rightarrow B$.  The following theorem gives an exact expression for the conditional expectation of $\Psi$ in terms of the underlying parameters $U$, $\Delta$, $p$, and $b$:

\begin{theorem} \label{yexp}
\begin{multline} \label{ey3}
    E[\Psi | X_D(u)=d] = b(1-b)p^u_d + b^2 +\\
     \sum_{S\subseteq U : u\in S} \quad \sum_{\Delta^0 \in \mathbb{N}^{\Delta}: |\Delta^0|\le S} b^{n-|S|+|\Delta^0|}(1-b)^{2|S|-|\Delta^0|} \cdot \\
     \left( \sum_{T\subseteq S-u: |T|=|\Delta^0|-1} \quad \sum_{\pi \in \Pi(T+u,\Delta^0) : \pi(u)=d} p^u_d \prod_{v\in T} p^v_{\pi(v)} \right. \\
     \left. + \sum_{T\subseteq S-u: |T|=|\Delta^0|} \quad \sum_{\pi \in \Pi(T,\Delta^0)} p^u_d \prod_{v\in T} p^v_{\pi(v)} \right)^2 \cdot \\
     [\rho(\Delta^0)]^{-1}(p^u_d)^{-1} \left(\sum_{T\subseteq S:|T|=|\Delta^0|} \quad \sum_{\pi \in \Pi(T,\Delta^0)} \quad \prod_{v\in T} p^v_{\pi(v)}\right)^{-1}
\end{multline}


\end{theorem}
\begin{proof}
At a high level, the conditional expectation of $\Psi$ can be expressed as:
\begin{equation*}
  \textnormal{E}[\Psi | X_D(u)=d] = \sum_{C\in \mathcal{C}} \textnormal{Pr}[X=C | X_D(u)=d] \Psi(C).
\end{equation*}

We calculate $\Psi$ for a configuration $C$ by finding the relative weight of indistinguishable configurations in which $u$ selects $d$.  The adversary observes some subset of the circuits.  If we match the users to circuits in some way that sends users with observed inputs to their own circuits, the result is an indistinguishable configuration.  Similarly, we can match circuits to destinations in any way that sends circuits on which the output has been observed to their actual destination in $C$.

The value of $\Psi(C)$ is especially simple if $u$'s input has been observed.  If the output has not also been observed, then $\Psi(C) = p^u_d$.  If the output has also been observed, then $\Psi(C) = 1$.

For the case in which $u$'s input has not been observed, we have to take into account the destinations of and observations on the other users.  Let $S\subseteq U$ be the set of users $s$ such that $C_I(s)=0$.  Note that $u\in S$.  Let $\Delta^0$ be the multiset of the destinations of circuits in $C$ on which the input has not been observed, but the output has.

Let $f_0(S,\Delta^0)$ be the probability that in a random configuration the set of unobserved inputs is $S$ and the set of observed destinations with no corresponding observed input is $\Delta^0$:
\begin{equation*}
f_0(S,\Delta^0) = b^{n-|S|+|\Delta^0|} (1-b)^{2|S|-|\Delta^0|} [\rho(\Delta^0)]^{-1} \sum_{T\subseteq S:|T|=|\Delta^0|} \ \sum_{\pi \in \Pi(T,\Delta^0)} \ \prod_{v\in T} p^v_{\pi(v)}.
\end{equation*}

Let $f_1(S,\Delta^0)$ be the probability that in a random configuration the set of unobserved inputs is $S$, the set of observed destinations with no corresponding observed input is $\Delta^0$, the output of $u$ is observed, and the destination of $u$ is $d$:
\begin{multline*}
f_1(S,\Delta^0) = b^{n-|S|+|\Delta^0|} (1-b)^{2|S|-|\Delta^0|} [\rho(\Delta^0)]^{-1} p^u_d \cdot \\
\sum_{T\subseteq S-u:|T|=|\Delta^0|-1} \ \sum_{\pi \in \Pi(T+u,\Delta^0) : \pi(u)=d} \ \prod_{v\in T} p^v_{\pi(v)}.
\end{multline*}

Let $f_2(S,\Delta^0)$ be the probability that in a random configuration the set of unobserved inputs is $S$, the set of observed destinations with no corresponding observed input is $\Delta^0$, the output of $u$ is unobserved, and the destination of $u$ is $d$:
\begin{multline*}
f_2(S,\Delta^0) = b^{n-|S|+|\Delta^0|} (1-b)^{2|S|-|\Delta^0|}  [\rho(\Delta^0)]^{-1} p^u_d \cdot \\
\sum_{T\subseteq S-u:|T|=|\Delta^0|} \ \sum_{\pi \in \Pi(T,\Delta^0)} \ \prod_{v\in T} p^v_{\pi(v)}.
\end{multline*}

Now we can express the posterior probability $\Psi(C)$ as:
\begin{equation} \label{yform}
  \Psi(C) = \frac{f_1(S,\Delta^0) + f_2(S,\Delta^0)}{f_0(S,\Delta^0)}.
\end{equation}

The expectation of $\Psi$ is a sum of the above posterior probabilities weighted by their probability.  The probability that the input of $u$ has been observed but the output hasn't is $b(1-b)$.  The probability that both the input and output of $u$ have been observed is $b^2$.  These cases are represented by the first two terms in Equation~\ref{ey3}.

When the input of $u$ has not been observed, we have an expression of the posterior in terms of sets $S$ and $\Delta^0$.  The numerator ($f_1+f_2$) of Equation~\ref{yform} itself actually sums the weight of every configuration that is consistent with $S$, $\Delta^0$, and the fact that the destination of $u$ is $d$.  However, we must divide by $p^u_d$, because we condition on the event $\{X_D(u)=d\}$.

These observations give us the final summation in Equation~\ref{ey3}.
\end{proof}

\subsection{ Simple approximation of conditional expectation}
The expression for the conditional expectation of $\Psi$ in Equation~\ref{ey3} is difficult to interpret.  It would be nice if we could find a simple approximation.  The probabilistic analysis in \citeN{onion-routing:pet2000} proposes just such a simplification by reducing it to only two cases: $i)$ the adversary observes the user's input and output and therefore identifies his destination, and $ii)$ the adversary doesn't observe these and cannot improve his \emph{a priori} knowledge.  The corresponding simplified expression for the expection is:
\begin{equation} \label{lower}
E[\Psi | X_D(u) = d] \approx b^2 + (1-b^2)p^u_d.
\end{equation}
This is a reasonable approximation if the final summation in Equation~\ref{ey3} is about $(1-b)p^u_d$.  This summation counts the case in which $u$'s input is not observed, and to achieve a good approximation the adversary must experience no significant advantage or disadvantage from comparing the users with unobserved inputs ($S$) with the discovered destinations ($\Delta^0$).

The quantity $(1-b)p^u_d$ does provide a lower bound on the final summation.  It may seem obvious that considering the destinations in $\Delta^0$ can only improve the accuracy of adversary's prior guess about $u$'s destination.  However, in some situations the posterior probability for the correct destination may actually be smaller than the prior probability.  This may happen, for example, when some user $v$, $v\neq u$, communicates with a destination $e$, $e\neq d$, and only $u$ is \emph{a priori} likely to communicate with $e$.  If the adversary observes the communication to $e$, it may infer that it is likely that $u$ was responsible and therefore didn't choose $d$.

It is true, however, that in expectation this probability can only increase.  Therefore Equation~\ref{lower} provides a lower bound on the anonymity metric.

The proof of this fact relies on the following lemma.  Let $\mathcal{E}$ be an event in some finite sample space $\Omega$.  Let $\mathcal{A}_1,\ldots ,\mathcal{A}_n$ be a set of disjoint events such that $\mathcal{E} \subseteq \bigcup_i \mathcal{A}_i$, and let $\mathcal{A}^j = \bigcup_{i=1}^j \mathcal{A}_i$.  Let $\mathcal{E}_i = \mathcal{E}\cap \mathcal{A}_i$.  Finally, let $Y(\omega) = \sum_i 1_{\mathcal{E}_i}(\omega) Pr[\mathcal{E}_i]/Pr[\mathcal{A}_i]$ (where $1_{\mathcal{E}_i}$ is the indicator function for $\mathcal{E}_i$). $Y(\omega)$ is thus the conditional probability $Pr[\mathcal{E} | \mathcal{A}_i]$, where $\omega \in \mathcal{E}_i$.

\begin{lemma} \label{minlem}
$Pr[\mathcal{E} | \mathcal{A}^n] \le E[Y | \mathcal{E}]$
\end{lemma}
\begin{proof}
\begin{displaymath}
\begin{array}{lll}
Pr[\mathcal{E} | \mathcal{A}^n] &=  \frac{Pr[\mathcal{E}]}{Pr[\mathcal{A}^n]}&\\
&= \frac{\left( \sum_i \frac{Pr[\mathcal{E}_i] \sqrt{Pr[\mathcal{A}_i]}}{\sqrt{Pr[\mathcal{A}_i]}} \right)^2}{Pr[\mathcal{A}^n] Pr[\mathcal{E}]}& \textrm{by a simple rewriting}\\
&\le \frac{\left( \sqrt{\sum_i \frac{(Pr[\mathcal{E}_i])^2}{Pr[\mathcal{A}_i]}}\sqrt{\sum_i Pr[\mathcal{A}_i]} \right)^2}{Pr[\mathcal{A}^n] Pr[\mathcal{E}]} & \textrm{by the Cauchy-Schwartz inequality}\\
&= \sum_i \frac{(Pr[\mathcal{E}_i])^2}{Pr[\mathcal{A}_i] Pr[\mathcal{E}]}\\
&= E[Y | \mathcal{E}]
\end{array}
\end{displaymath}
\end{proof}

\begin{theorem} \label{thm:lwrbnd}
$E[\Psi | X_D(u)=d] \ge  b^2 + (1-b^2)p^u_d$
\end{theorem}
\begin{proof}
As described in the proof of Theorem~\ref{yexp}:
\begin{equation*}
E[\Psi | X_D(u)=d] = b^2 + b(1-b)p^u_d + (1-b)E[\Psi | X_D(u)=d \land X_I(u)=0].
\end{equation*}

To apply Lemma~\ref{minlem}, take the set of configurations $\mathcal{C}$ to be the sample space $\Omega$.  Take $\{X_D(u)=d\}$ to be the event $\mathcal{E}$.  Take the indistinguishability equivalence relation to be the sets $\mathcal{A}_i$.  Finally, take $\Psi$ to be $Y$.  Then the lemma shows that $E[\Psi | X_D(u)=d \land X_I(u)=0]\ge p^u_d$.
\end{proof}

\subsection{Worst-case Anonymity}
To examine the accuracy of our approximation, we look at how large the final summation in Equation~\ref{ey3} can get as the users' destination distributions vary.  Because this is the only term that varies with the other user distributions, this will also provide a worst-case guarantee on expected anonymity metric.  Our results will show that, in the limit as the number of users grows, the worst case can occur when the users other than $u$ act as differently from $u$ as possible by always visiting the destination $u$ is otherwise least likely to visit.  Less obviously, we show that the limiting maximum can also occur when the users other than $u$ always visit $d$.  This happens because it makes the adversary observe destination $d$ often, causing him to suspect that $u$ chose $d$.  Our results also show that the worst-case expectation is about $b+(1-b)p^u_d$, which is significantly worse than the simple approximation above.

As the first step in finding the maximum of Equation~\ref{ey3} over $(p^v)_{v\neq u}$, we observe that it is obtained when every user $v\neq u$ chooses only one destination $d_v$, \emph{i.e.} $p^v_{d_v}=1$ for some $d_v\in \Delta$.

\begin{lemma} \label{vertex}
  A maximum of $E[\Psi | X_D(u)=d]$ over $(p^v)_{v\neq u}$ must occur when, for all $v\neq u$, there exists some $d_v\in \Delta$ such that $p^v_{d_v}=1$.
\end{lemma}
\begin{proof}
Take some user $v\neq u$ and two destinations $e,f\in \Delta$.  Assign arbitrary probabilities in $p^v$ to all destinations except for $f$, and let $\zeta=1-\sum_{\delta \neq e,f}p^v_{\delta}$.  Then $p^v_f = \zeta -p^v_e$.  Consider $E[\Psi | X_D(u)=d]$ as a function of $p^v_e$.  The terms $t_i$ of Equation~\ref{ey3} that correspond to any fixed $S$ and $\Delta^0$ are of the following general form, where $c_1^i, c_2^i, c_3^i, c_4^i, c_5^i, c_6^i \ge 0$:
\begin{displaymath}
t_i = \frac{(c_1^i p^v_e + c_2^i({\zeta}-p^v_e)+c_3^i)^2}{c_4^i p^v_e + c_5^i (\zeta-p^v_e) + c_6^i}.
\end{displaymath}

This is a convex function of $p^v_e$:
\begin{displaymath}
D^2_{p_e^v} t_i = \frac{2(c_3^i(c_4^i-c_5^i)+c_2^i(c_4^i \zeta+c_6^i)-c_1^i(c_5^i \zeta+c_6^i))^2}{(c_5^i(\zeta-p^v_e)+c_4^i p^v_e+c_6^i)^3} \ge 0.
\end{displaymath}

The leading two terms of $E[\Psi | X_D(u)=d]$ are constant in $p^v$, and the sum of convex functions is a convex function, so $E[\Psi | X_D(u)=d]$ is convex in $p^v_e$.  Therefore, a maximum of $E[\Psi | X_D(u)=d]$ must occur when $p^v_e\in \{0,1\}$.\hfill\end{proof}

Order the destinations $d=d_1,\ldots ,d_{|\Delta|}$ such that $p^u_{d_i}\ge p^u_{d_{i+1}}$ for $i>1$.  The following lemma shows that we can further restrict ourselves to distribution vectors in which, for every user except $u$, the user either always chooses $d$ or always chooses $d_{|\Delta|}$.

\begin{lemma} \label{ef}
A maximum of $E[\Psi | X_D(u)=d]$ must occur when, for all users $v$, either $p^v_{d_1}=1$ or $p^v_{d_{|\Delta|}}=1$.
\end{lemma}
\begin{proof}
Assume, following Lemma~\ref{vertex}, that $(p^v)_{v\neq u}$ is an extreme point of the set of possible distribution vectors.

Equation~\ref{ey3} groups configurations first by the set $S$ with unobserved inputs and second by the observed destinations $\Delta^0$.  Instead, group configurations first by $S$ and second by the set $T\subseteq S$ with observed outputs.  Because every user except $u$ chooses a destination deterministically, $\Psi$ only depends on the sets $S$ and $T$.  Let $\Psi_1(S,T)$ be this value.
\begin{equation} \label{ey4}
  \begin{array}{ll}
  E[\Psi | X_D(u)=d] = &b(1-b)p^u_d + b^2 +\\
         & \sum_{S : u\in S} \sum_{T : T\subseteq S} b^{n-|S|+|T|}(1-b)^{2|S|-|T|} \Psi_1(S,T).
  \end{array}
\end{equation}

Select two destinations $d_i,d_j, 1<i<j$.  We break up the sum in Equation~\ref{ey4} and show that, for every piece, the sum can only be increased by changing $(p^v)_v$ so that any user that always chooses $d_i$ always chooses $d_j$ instead.

Fix $S\subseteq U$ such that $u\in S$.  Let $S_i, S_j\subseteq S$ be such that $p^s_{d_i}=1$ if and only if $s\in S_i$, and $p^s_{d_j}=1$ if and only if $s\in S_j$.  Fix $T'\subseteq S\backslash S_i\backslash S_j$ and some $t\ge |T'|$.

Let $f(S,T')$ be the sum of terms in Equation~\ref{ey4} that are indexed by $S$ and some $T$ such that $|T|=t$ and $T\supseteq T'$.  To calculate $f(S,T')$, group its terms by the number $t_{d_i}$ of users $v$ in $T$ such that $X_D(v)=d_i$.  Let $t_e$ be the number for these terms of users $v$ in $T'$ such that $X_D(v)=e$, $e\in \Delta \backslash \{d_i,d_j\}$.  The number $t_{d_j}$ of users $v$ such that $X_D(v)=d_j$ for these terms is then $t-\sum_{e\in \Delta - d_j} t_e$.  Let $s_e$ be the number of users $v$ in $S-u$ such that $X_D(v)=e$. The number of terms in $f(S,T')$ with a given $t_{d_i}$ is then
\begin{equation*}
\binom{s_{d_i}}{t_{d_i}} \binom{s_{d_j}}{t_{d_j}}.
\end{equation*}
For each of these terms, $\Psi_1$ is the same.  To calculate it, let $f_{\delta}$ be the number of configurations that yield the given $S$ and $(t_e)_{e\in \Delta}$ and are such that $u$'s output is observed with destination $\delta$:
\begin{equation*}
f_{\delta}(t_{d_i}) = \binom{s_{\delta}}{t_{\delta}-1} \prod_{e\in \Delta - \delta} \binom{s_{e}}{t_{e}},
\end{equation*}
and let $f_0$ be the number of configurations that yield the same $S$ and $(t_e)_{e\in \Delta}$ and are such that $u$'s output is unobserved:
\begin{equation*}
f_0(t_{d_i}) =  \prod_{e\in \Delta} \binom{s_e}{t_e}.
\end{equation*}
Then the posterior probability given $S$ and $(t_e)_{e\in \Delta}$ is
\begin{equation*}
\frac{p^u_d \left(f_d(t_{d_i}) + f_0(t_{d_i}) \right)}{\sum_{\delta \in \Delta} p^u_{\delta} f_{\delta}(t_{d_i}) + f_0(t_{d_i})}.
\end{equation*}
Therefore, letting $m= t-\sum_{e \in \Delta\backslash \{d_i,d_j\}} t_e$,
\begin{align*}
  f(S,T') = &b^{n-|S|+t}(1-b)^{2|S|-t}\sum_{t_{d_i}=0}^m \binom{s_{d_i}}{t_{d_i}} \binom{s_{d_j}}{m-t_{d_i}} \frac{p^u_d \left(f_d(t_{d_i}) + f_0(t_{d_i}) \right)}{\sum_{\delta \in \Delta} p^u_{\delta} f_{\delta}(t_{d_i}) + f_0(t_{d_i})}.
\end{align*}

The binomial coefficients of $f_{\delta}$ and $f_0$ in the numerator and denominator largely cancel, and the whole expression can be simplified to
\begin{equation*}
  f(S,T') = \alpha \sum_{t_{d_i}=0}^m  \binom{s_{d_i}}{t_{d_i}} \binom{s_{d_j}}{m-t_{d_i}} \frac{(s_{d_i}+1-t_{d_i})(s_{d_j}+1-m+t_{d_i})}{\left( \begin{array}{l} p^u_{d_i}(s_{d_i}+1)(s_{d_j}+1-m+t_{d_i}) +\\ p^u_{d_j}(s_{d_j}+1)(s_{d_i}+1-t_{d_i}) +\\ (s_{d_i}+1-t_{d_i})(s_{d_j}+1-m+t_{d_i})\beta \end{array} \right)}
\end{equation*}
for some $\alpha,\beta \ge 0$.

This can be seen as the weighted convolution of binomial coefficients.  Unfortunately, there is no obvious way to simplify the expression any further to find the maximum as we trade off $s_{d_i}$ and $s_{d_j}$.  There is a closed-form sum if the coefficient of the binomial product is a fixed-degree polynomial, however.  Looking at the coefficient, we can see that it is concave.
\begin{displaymath}
  \begin{array}{lll}
    c_{t_{d_i}} &= &\frac{(s_{d_i}+1-t_{d_i})(s_{d_j}+1-m+t_{d_i})}{p^u_{d_i}(s_{d_i}+1)(s_{d_j}+1-m+t_{d_i}) + p^u_{d_j}(s_{d_j}+1)(s_{d_i}+1-t_{d_i}) + (s_{d_i}+1-t_{d_i})(s_{d_j}+1-m+t_{d_i})\beta}.\\
& & \\
    D^2_{t_{d_i}} c_{t_{d_i}} &= &-\frac{\left( \begin{array}{l} 2((s_{d_i}+1)(s_{d_j}+1)(2+s_{d_i}+s_{d_j}-m)^2 p^u_{d_i}p^u_{d_j}+ \\ b((s_{d_i}+1)(s_{d_j}+1+t_{d_i}-m)^3 p^u_{d_i} +  (s_{d_j}+1)(s_{d_i}+1-t_{d_i})^3 p^u_{d_j})) \end{array} \right) }
{((s_{d_j}+1+t_{d_i}-m)(b(s_{d_i}+1-t_{d_i})+p^u_{d_i}(s_{d_i}+1)) + (s_{d_j}+1)(s_{d_i}+1-t_{d_i}) p^u_{d_j})^3}\\
    &\le &0.
  \end{array}
\end{displaymath}

We can use this fact to bound the sum above by replacing $c_{t_{d_i}}$ with a line tangent at some point $i_0$.  Call this approximation $\tilde{f}$.  Holding $s_{d_i}+s_{d_j}$ constant, this approximation is in fact equal at $s_{d_i}=0$ because the sum has only one term.  Then, if $s_{d_i}=0$ still maximizes the sum, the theorem is proved.  Let $c'_{i_0} = D_{t_{d_i}} c_{t_{d_i}} \big|_{t_{d_i}=i_0}$.

\begin{eqnarray*}
  f(S,T') &\le& \sum_{t_{d_i}=0}^m \binom{s_{d_i}}{t_{d_i}} \binom{s_{d_j}}{m-t_{d_i}}(c'_{i_0}(t_{d_i}-i_0)+c_{i_0})\\
          &=& \binom{s_{d_i}+s_{d_j}}{m} \left(c_{i_0} + c'_{i_0} \frac{m \cdot s_{d_i}}{s_{d_i}+s_{d_j}} - c'_{i_0} i_0 \right)\\
          &=& \tilde{f}(S,T').
\end{eqnarray*}

The linear approximation will be done around the point $i_0 = m \cdot s_{d_i}/(s_{d_i}+s_{d_j})$.  This results in a simple form for the resulting approximation, and also the mass of the product of binomial coefficients concentrates around this point.  Set $\nu = s_{d_i}+s_{d_j}$ to examine the tradeoff between $s_{d_i}$ and $s_{d_j}$.

\begin{eqnarray*}
  \tilde{f}(S,T') &=& \binom{\nu }{m} \left( c_{\frac{m \cdot s_{d_i}}{\nu }} \right)\\
                  &=& \binom{\nu }{m} \frac{((\nu -s_{d_i})(\nu -m)+\nu )((s_{d_i}+1)\nu -m \cdot s_{d_i})}{\left( \begin{array}{l}p^u_{d_i} \nu (s_{d_i}+1)((\nu -s_{d_i})(\nu -m)+\nu ) + \\ p^u_{d_j}\nu (\nu -s_{d_i}+1)(\nu +s_{d_i}(\nu -m)) + \\ \beta((s_{d_i}+1)\nu -m \cdot s_{d_i})((\nu -s_{d_i})(\nu -m)+\nu ) \end{array} \right)}.
\end{eqnarray*}

Lemma~\ref{lem:d2f} in the Appendix shows that $\tilde{f}$ is convex in $s_{d_i}$.  Thus, the maximum of $\tilde{f}$ must exist at $s_{d_i}=0$ or $s_{d_i} = \nu$.  Observe that when $s_{d_i}=0$,
\begin{equation*}
\tilde{f} = \binom{\nu}{m} \frac{1-m+\nu}{p_{d_j}(1+\nu) + \beta(1 -m+\nu) +p_{d_i} (1-m+\nu )}
\end{equation*}
and when $s_{d_i} = \nu$
\begin{equation*}
\tilde{f} = \binom{\nu}{m} \frac{1-m+\nu}{p_{d_j}(1-m+\nu) + \beta(1 -m+\nu) +p_{d_i} (1+\nu )}.
\end{equation*}
Therefore, because $p_{d_i}\ge p_{d_j}$, $\tilde{f}$ is larger when $s_{d_i}=0$.  As stated, this implies that $f$ itself is maximized when $s_{d_i}=0$.\hfill

\end{proof}

Therefore, in looking for a maximum we can assume that every user except $u$ either always visits $d$ or always visits $d_{|\Delta|}$.  To examine how anonymity varies with the number of users in each category, we derive an asymptotic estimate for large $n$.  A focus on large $n$ is reasonable because anonymity networks, and onion routing in particular, are understood to have the best chance at providing anonymity when they have many users.  Furthermore, Tor is currently used by an estimated 500,000 people.

Let $\alpha = \{v\neq u : p^v_d = 1 \}/(n-1)$ be the fraction of users that always visit $d$.  Theorem~\ref{thm:de} gives an asymptotic estimate for the expected posterior probability given a constant $\alpha$.  It shows that, in the limit, the maximum expected posterior probability is obtained when all users but $u$ always visit $d$ or when they always visit $d_{|\Delta|}$.
\begin{theorem} \label{thm:de}
Assume that, for all $v\neq u$, either  $p^v_d=1$ or $p^v_{d_{|\Delta|}}=1$.  Then, if $\alpha = 0$,
\begin{equation*}
E[\Psi | X_D(u)=d] = b(1-b)p^u_d + b^2 + (1-b) \left( b+\frac{(1-b)^2 p^u_d }{1 - b + p^u_{d_{|\Delta|}} b} \right) + O\left(\sqrt{\frac{\log(n)}{n}} \right),
\end{equation*}
if $0<\alpha<1$
\begin{equation*}
E[\Psi | X_D(u)=d] = b(1-b)p^u_d + b^2 + (1-b) \frac{p^u_d }{1-b + p^u_d b + p^u_{d_{|\Delta|}}b} + O\left(\sqrt{\frac{\log(n)}{n}} \right),
\end{equation*}
and, if $\alpha=1$,
\begin{equation*}
E[\Psi | X_D(u)=d] = b(1-b)p^u_d + b^2 + (1-b) \frac{p^u_d}{1-b + p^u_d b} + O\left(\sqrt{\frac{\log(n)}{n}} \right).
\end{equation*}
\end{theorem}
\begin{proof}
Let $n_e = \alpha(n-1)$ and $n_f = (1-\alpha)(n-1)$.  The expected posterior probability can be given in the following variation on Equation~\ref{ey4}:
\begin{equation} \label{eq:yasymp}
\begin{split}
  &E[\Psi | X_D(u)=d] = b(1-b)p^u_d + b^2 + (1-b) \cdot\\
  &\qquad \sum_{e=0}^{n_e} \binom{n_e}{e} (1-b)^e b^{n_e-e} \sum_{f=0}^{n_f} \binom{n_f}{f} (1-b)^f b^{n_f-f} \cdot\\
  &\qquad \sum_{j=0}^f \binom{f}{j} b^j (1-b)^{f-j} \sum_{k=0}^{e} \binom{e}{k} b^k (1-b)^{e-k}\cdot \\
  &\qquad \qquad \left[ b \Psi_2(e,f,j,k+1) + (1-b)\Psi_2(e,f,j,k)  \right].
\end{split}
\end{equation}
Here $\Psi_2(e,f,j,k)$ is the value of $\Psi$ when the users with unobserved inputs consist of $u$, $e$ users $v\neq u$ with $p^v_d = 1$, and $f$ users $v\neq u$ with $p^v_{d_{|\Delta|}}=1$; and the users with unobserved inputs and observed outputs consist of $k$ users $v$ with $X_D(v)=d$ and $j$ users $v$ with $X_D(v)=d_{|\Delta|}$.  Given such a configuration, the number of indistinguishable configurations in which $u$ has observed destination $d$ is $\binom{e}{k-1} \binom{f}{j}$, the number of indistinguishable configurations in which $u$ has observed destination $d_{|\Delta|}$ is $\binom{e}{k} \binom{f}{j-1}$, and the number of indistinguishable configuration in which $u$ has an unobserved destination is $\binom{e}{k} \binom{f}{j}$.  Thus, we can express $\Psi_2$ as
\begin{equation*}
\Psi_2(e,f,j,k) =  \frac{p^u_d \binom{e}{k-1}\binom{f}{j} + p^u_d \binom{e}{k}\binom{f}{j}}{p^u_d \binom{e}{k-1}\binom{f}{j} + p^u_{d_{|\Delta|}} \binom{e}{k}\binom{f}{j-1} + \binom{e}{k}\binom{f}{j}}.
\end{equation*}
The binomial coefficients largely cancel, and so we can simplify this equation to
\begin{equation*}
\Psi_2(e,f,j,k) = \frac{p^u_d (e+1)(f-j+1)}{p^u_d k (f-j+1) + p^u_{d_{|\Delta|}} j(e-k+1) + (e-k+1)(f-j+1) } \label{eq:y}.
\end{equation*}

We observe that $j$ and $k$ are binomially distributed.  Therefore, by the Chernoff bound, they concentrate around their means as $e$ and $f$ grow.  Let $\mu_1 = fb$ be the mean of $j$ and $\mu_2 = eb$  be the mean of $k$.  We can approximate the tails of the sums over $j$ and $k$ in Equation~\ref{eq:yasymp} and sum only over the central terms:
\begin{equation} \label{eq:yasymp2}
\begin{split}
  &E[\Psi | X_D(u)=d] = b(1-b)p^u_d + b^2 +\\
  &\qquad (1-b) \sum_{e=0}^{n_e} \binom{n_e}{e} (1-b)^e b^{n_e-e} \sum_{f=0}^{n_f} \binom{n_f}{f} (1-b)^f b^{n_f-f} \cdot\\
  &\qquad \bigg[ O\left(\exp(-2 c_1)\right) + O\left( \exp(-2 c_2)\right) +\\
  &\qquad \qquad \sum_{j: |j-\mu_1| < \sqrt{c_1 f}} \binom{f}{j} b^j (1-b)^{f-j} \sum_{k: |k-\mu_2|<\sqrt{c_2 e}} \binom{e}{k} b^k (1-b)^{e-k}\cdot\\
  &\qquad \qquad\left( b \Psi_2(e,f,j,k+1) + (1-b)\Psi_2(e,f,j,k) \right) \bigg].  
\end{split}
\end{equation}

As $j$ and $k$ concentrate around their means, $\Psi_2$ will approach its value at those means.  Let
\begin{equation*}
\varepsilon_1(j,k,u) = \Psi_2(e,f,j,k+u) - \Psi_2(e,f,\mu_1,\mu_2+u)
\end{equation*}
be the difference of $\Psi_2$ from its value at $\mu_1$ and $\mu_2$+u, where $u\in \{0,1\}$ indicates if $u$'s output is observed.

$\Psi_2$ is non-increasing in $j$ and is non-decreasing in $k$:
\begin{align*}
D_j \Psi_2 &= -\frac{(1+e) (1+f) (1+e-k) p^u_{d_{|\Delta|}} p^u_d}{\left( \begin{array}{l}(1+f) (1+e-k) p^u_{d_{|\Delta|}} + \\(1+f-j-u) (p^u_d (e+1) +(1-p^u_d-p^u_{d_{|\Delta|}})(1+e-k)) \end{array} \right)^2}\\
&\le 0.\\
D_k \Psi_2 &= \frac{(1+e) (1+f-j) p^u_d  (p^u_{d_{|\Delta|}}(1+f)+(1-p^u_d-p^u_{d_{|\Delta|}})(1+f-j))}{\left( \begin{array}{l} (1+f) (1+e-k-u)p^u_{d_{|\Delta|}} +\\ (1+f-j) ((1+e) p^u_d +(1+e-k-u) (1-p^u_d-p^u_{d_{|\Delta|}}) ) \end{array} \right)^2}\\
&\ge 0.
\end{align*}

Because the signs of these derivatives are constant, the magnitude of $\varepsilon_1$ is largest when $j$ and $k$ are as large or as small as possible.  We can therefore bound the magnitude of $\varepsilon_1$ with
\begin{align*}
&\max_{\substack{\sigma\in \{-1,1\}\\ u\in \{0,1\} }}\left(\left| \varepsilon_1 \left(\mu_1 + \sigma \sqrt{c_1 f},\mu_2 + \sigma \sqrt{c_2 e},u \right) \right| \right)\\
&\qquad = \max_{\substack{\sigma\in \{-1,1\}\\ u\in \{0,1\} }} \left| \Psi_2(e,f,\mu_1 + \sigma \sqrt{c_1 f},\mu_2 + \sigma \sqrt{c_2 e}+u) - \Psi_2(e,f,\mu_1,\mu_2+u) \right|\\
&\qquad = O\left( \sqrt{c_1/f} \right) + O\left( \sqrt{c_2/e} \right),
\end{align*}
where the second line follows from a simple expansion of $\Psi_2$ according to Equation~\ref{eq:y}.  We use this estimate to approximate the value of $\Psi_2$:
\begin{align*}
\Psi_2(e,f,j,k+u) &= \Psi_2(e,f,\mu_1,\mu_2+u) + \varepsilon_1(j,k,u)\\
&= \Psi_2(e,f,\mu_1,\mu_2+u) + O\left( \sqrt{c_1/f} \right) + O\left( \sqrt{c_2/e} \right).
\end{align*}
We set $c_1 = \log(f)/4$ and $c_2 = \log(e)/4$, and then Equation~\ref{eq:yasymp2} becomes
\begin{equation} \label{eq:yasymp2.5}
\begin{split}
  &E[\Psi | X_D(u)=d] = b(1-b)p^u_d + b^2 + \\
  &\qquad (1-b) \sum_{e=0}^{n_e} \binom{n_e}{e} (1-b)^e b^{n_e-e} \sum_{f=0}^{n_f} \binom{n_f}{f} (1-b)^f b^{n_f-f} \cdot\\
  &\qquad \bigg[ b \Psi_2(e,f,\mu_1,\mu_2+1) + (1-b)\Psi_2(e,f,\mu_1,\mu_2) +\\
  &\qquad \qquad O\left( \sqrt{\log(f)/f} \right) + O\left( \sqrt{\log(e)/e} \right) \bigg].
\end{split}
\end{equation}

$e$ and $f$ in this expression are binomially distributed.  Let $\mu_3 = n_e(1-b)$ be the mean of $e$ and $\mu_4 = n_f(1-b)$ be the mean of $f$.  By applying the Chernoff bound to the sum over $e$, setting the tails to start at $\min(b,1-b)n_e/2$ from $\mu_3$, we can see that
\begin{equation*}
\sum_{e=0}^{n_e} \binom{n_e}{e} (1-b)^e b^{n_e-e} \sum_{f=0}^{n_f} \binom{n_f}{f} (1-b)^f b^{n_f-f} O\left(\sqrt{\log(e)/e}\right) = O\left(\sqrt{\log(n_e)/n_e} \right).
\end{equation*}
We can similarly show that
\begin{equation*}
\sum_{e=0}^{n_e} \binom{n_e}{e} (1-b)^e b^{n_e-e} \sum_{f=0}^{n_f} \binom{n_f}{f} (1-b)^f b^{n_f-f} O\left(\sqrt{\log(f)/f} \right) = O\left(\sqrt{\log(n_f)/n_f} \right).
\end{equation*}

For the remaining terms inside both sums, approximate the sums over $e$ and $f$ using the Chernoff bound by setting the tails to be those terms more than $\sqrt{c_3 n_e}$ from $\mu_3$ and more than $\sqrt{c_4 n_f}$ from $\mu_4$, respectively. This yields
\begin{equation} \label{eq:yasymp3}
\begin{split}
  &E[\Psi | X_D(u)=d] = b(1-b)p^u_d + b^2 +\\
  &\qquad O\left( \left( log(n_e)/n_e \right)^{-1/2} \right) + O\left( \left( log(n_f)/n_f \right)^{-1/2} \right) + O\left( e^{-2 c_3} \right) + O\left( e^{-2 c_4} \right) + \\
  &\qquad (1-b) \sum_{e: |e-\mu_3|<\sqrt{c_3 n_e}} \binom{n_e}{e} (1-b)^e b^{n_e-e} \sum_{f: |f-\mu_4| < \sqrt{c_4 n_f}} \binom{n_f}{f} (1-b)^f b^{n_f-f} \cdot\\
  &\qquad \qquad \left[ b \Psi_2(e,f,\mu_1,\mu_2+1) + (1-b)\Psi_2(e,f,\mu_1,\mu_2) \right].
\end{split}
\end{equation}

As $e$ and $f$ concentrate around their means, $\Psi_2$ will approach its value at those means.  Let
\begin{equation*}
\varepsilon_2(e,f,u) = \Psi_2(e,f,\mu_1,\mu_2+u) - \Psi_2(\mu_3,\mu_4,\mu_1,\mu_2+u)
\end{equation*}
be the difference of $\Psi_2$ from its value at $e=\mu_3$ and $f=\mu_4$, $u\in \{0,1\}$.  $\Psi_2(e,f,\mu_1,\mu_2)$ in non-decreasing with respect to $e$:
\begin{align*}
D_e \Psi_2(e,f,\mu_1,\mu_2) &= \frac{(1+(1-b) f) b p^u_d  ((f+1)(1-p^u_d) - fb(1-p^u_d-p^u_{d_{|\Delta|}}))}{\left( \begin{array}{l} (1+(1-b)f)(1+(1-b)e) + \\ (1+(1-b) f)(be)p^u_d + \\ b f (1+(1-b) e+u) p^u_{d_{|\Delta|}} \end{array} \right)^2}\\
&\ge 0.
\end{align*}
$\Psi_2(e,f,\mu_1,\mu_2+1)$ is non-increasing with respect to $e$:
\begin{align*}
D_e \Psi_2(e,f,\mu_1,\mu_2) &= \frac{(1+(1-b) f) (1-b) p^u_d  (fb(1-p^u_{d_{|\Delta|}}-p^u_d)-(f+1)(1-p^u_d))}{\left( \begin{array}{l} ((1-b)f)(1+(1-b)e) + \\ (1+(1-b) f)(be+1)p^u_d + \\ b f ((1-b) e) p^u_{d_{|\Delta|}} \end{array} \right)^2}\\
&\le 0.
\end{align*}
$\Psi_2(e,f,\mu_1,\mu_2+u)$, $u\in \{0,1\}$, is non-increasing with respect to $f$:
\begin{align*}
D_f \Psi_2(e,f,\mu_1,\mu_2+u) &= \frac{ - b(1+e)(1+(1-b) e+u) p^u_d p^u_{d_{|\Delta|}} }{\left( \begin{array}{l} (1+(1-b)f)(1+(1-b)e+u) +\\  (1+(1-b)f)(be+u)p^u_d + \\ bf(1+(1-b)e+u)p^u_{d_{|\Delta|}} \end{array} \right)^2}\\
&\le 0.
\end{align*}

Therefore, the magnitude of $\varepsilon_2$ is largest when $e$ and $f$ are as large or as small as possible.  We can therefore estimate the magnitude of $\varepsilon_2$ with
\begin{align*}
\max_{ \substack{\sigma\in\{-1,1\}\\ u\in \{0,1\} }} \left( \left| \varepsilon_2 \left(\mu_3+\sigma\sqrt{c_3 n_e},\mu_4+\sigma\sqrt{c_4 n_f},u \right) \right| \right).
\end{align*}
If $n_e, n_f \neq 0$,
\begin{align*}
\varepsilon_2 \left(\mu_3+\sigma\sqrt{c_3 n_e},\mu_4+\sigma\sqrt{c_4 n_f},u \right) = &\Psi_2(\mu_3+\sigma\sqrt{c_3 n_e},\mu_4+\sigma\sqrt{c_4 n_f},\mu_1,\mu_2+u) - \\
&\Psi_2(\mu_3,\mu_4,\mu_1,\mu_2+u)\\
= &O\left( \sqrt{c_3/n_e} \right) + O\left( \sqrt{c_4/n_f} \right) .
\end{align*}
If $n_e = 0$, which occurs when $\alpha = 0$,
\begin{align*}
\varepsilon_2 \left(0,\mu_4+\sigma\sqrt{c_4 n_f},u \right) &= \Psi_2(0,\mu_4+\sigma\sqrt{c_4 n_f},\mu_1,u) - \Psi_2(0,\mu_4,\mu_1,u)\\
&= O\left( \sqrt{c_4/n_f} \right).
\end{align*}
If $n_f = 0$, which occurs when $\alpha = 1$, the final term becomes
\begin{align*}
\varepsilon_2 \left(\mu_3+\sigma\sqrt{c_3 n_e},0,u \right) &= \Psi_2(\mu_3+\sigma\sqrt{c_3 n_e},0,0,\mu_2+u) - \Psi_2(\mu_3,0,0,\mu_2+u)\\
&= O\left( \sqrt{c_3/n_e} \right).
\end{align*}
These asymptotic estimates of $\varepsilon_2$ follow from a simple expansion of $\Psi_2$ according to Equation~\ref{eq:y}.

We use them estimate to approximate the value of $\Psi_2$ as $e$ and $f$ grow:
\begin{align*}
\Psi_2(e,f,\mu_1,\mu_2+u) &= \Psi_2(\mu_3,\mu_4,\mu_1,\mu_2+u) + \varepsilon_2(e,f,u)\\
&= \Psi_2(\mu_3,\mu_4,\mu_1,\mu_2+u) + O\left( \sqrt{c_3/n_e} \right) + O\left( \sqrt{c_4/n_f} \right).
\end{align*}
We set $c_3 = \log(n_e)/4$ and $c_4 = \log(n_f)/4$, and then Equation~\ref{eq:yasymp3} becomes
\begin{equation} \label{eq:yasymp4}
\begin{split}
&E[\Psi | X_D(u)=d] = b(1-b)p^u_d + b^2 + \\
&\qquad (1-b) \left[ b \Psi_2(\mu_3,\mu_4,\mu_1,\mu_2+1) + (1-b) \Psi_2(\mu_3,\mu_4,\mu_1,\mu_2) \right] + \\
&\qquad O\left( \left( log(n_e)/n_e \right)^{-1/2} \right) + O\left( \left( log(n_f)/n_f \right)^{-1/2} \right).
\end{split}
\end{equation}

Finally, we must estimate $\Psi_2(\mu_3,\mu_4,\mu_1,\mu_2+u)$, $u\in \{0,1\}$.  Assume that $0<\alpha<1$ and thus that $n_e = \alpha (n-1)$ and $n_f = (1-\alpha)(n-1)$ are both increasing with $n$.  Then
\begin{align*}
\Psi_2(\mu_3,\mu_4,\mu_1,\mu_2+u) &= \Psi_2((1-b)n_e,(1-b)n_f,b(1-b)n_f,b(1-b)n_e+u)\\
&= \frac{p^u_d (1-b)^3 n_e n_f + c_1 n_e + c_2 n_f + c_3 }{\left( \begin{array}{l} ((1-b)^4 + p^u_d (1-b)^3 b + p^u_{d_{|\Delta|}}(1-b)^3 b) n_e n_f +\\ c_4 n_e + c_5 n_f + c_6 \end{array} \right)}\\
&= \frac{p^u_d }{1-b + p^u_d b + p^u_{d_{|\Delta|}}b} + O(1/n_e) + O(1/n_f) + O(1/(n_e n_f)),
\end{align*}
where $c_1,\ldots ,c_6$ are some values constant in $n_e$ and $n_f$.
When $\alpha=0$, then $n_e=0$, and the estimate becomes
\begin{align*}
\Psi_2(\mu_3,\mu_4,\mu_1,\mu_2+u) &= \Psi_2(0,(1-b)n_f,b(1-b)n_f,u)\\
&= \frac{p^u_d (1-b) n_f + c_1}{ ((1-u)(1-b) + p^u_d u (1-b) + p^u_{d_{|\Delta|}} (1-u)b ) n_f + c_2}\\
&= \frac{p^u_d (1-b)}{ ((1-u)(1-b) + p^u_d u (1-b) + p^u_{d_{|\Delta|}} (1-u)b ) } + O(1/n_f),
\end{align*}
where $c_1,c_2$ are some values constant in $n_f$.
When $\alpha=1$, then $n_f=0$, and the estimate becomes
\begin{align*}
\Psi_2(\mu_3,\mu_4,\mu_1,\mu_2+u) &= \Psi_2((1-b)n_e,0,0,b(1-b)n_e+u)\\
&= \frac{p^u_d n_e + c_1}{((1-b) + p^u_d b)n_e +c_2} \\
&= \frac{p^u_d}{1-b + p^u_d b} + O(1/n_e),
\end{align*}
where $c_1,c_2$ are some values constant in $n_e$.

Inserting these estimates for $\Psi_2(\mu_3,\mu_4,\mu_1,\mu_2+u)$ into Equation~\ref{eq:yasymp4} yields the theorem.
\end{proof}

It follows from this theorem that the worst case anonymity over user distributions occurs either when all users always visit $d_{|Delta|}$ or when all users always visit $d$.
\begin{corollary} \label{cor:worst}
$\lim_{n\rightarrow \inf} E[\Psi | X_D(u)=d]$ is maximized either at $\alpha=0$ or at $\alpha=1$.
\end{corollary}
\begin{proof}
The case $\alpha=1$ is larger in the limit than the case where $0<\alpha < 1$, by Thm.~\ref{thm:de}, because
\begin{equation*}
 \frac{p^u_d }{1-b + p^u_d b + p^u_{d_{|\Delta|}}b} \le \frac{p^u_d}{1-b + p^u_d b}.
\end{equation*}
\end{proof}

The case $\alpha=1$ is the worst case only when
\begin{equation*}
p^u_{d_{|\Delta|}} \ge \frac{(1-b)(1-p^u_d)^2}{p^u_d(1+b)-b}.
\end{equation*}
This happens when $p^u_d\ge 1/2$ and $p^u_{d_{|\Delta|}}$ is near $1-p^u_d$. That is, if the user is likely to visit $d$ and the other users can't distinguish themselves too much, then it is worst to have them always visit $d$ because the adversary will blame $u$.

However, we would expect $p^u_{d_{|\Delta|}}$ to be small because it is at most $1/|\Delta|$. In this case the worst-case limiting distribution has $\alpha = 0$, that is, it is worst when the other users always act very different from $u$ by visiting $d_{|\Delta|}$. Then the expected assigned probability is about $b + (1-b)p^u_d$. This is equal to the lower bound on the anonymity metric when the adversary controls a fraction $\sqrt{b}$ of the network.

\section{Typical Distributions} \label{sec:typdist}
It is unlikely that users of onion routing will ever find themselves in the worst-case situation.  The necessary distributions just do not resemble what we expect user behavior to be like in any realistic use of onion routing.  Our worst-case analysis may therefore be overly pessimistic.  To get some insight into the anonymity that a typical user of onion routing can expect, we consider a more realistic set of users' destination distributions in which each user selects a destination from a common Zipfian distribution.  This model of user behavior is used by Shmatikov and Wang~\shortcite{ShWa-Relationship} to analyze relationship anonymity in mix networks and is motivated by observations that the popularity of sites on the web follows a Zipfian distribution.

Let each user select his destination from a common Zipfian distribution $p$:  $p_{d_i} = 1/(\mu i^s)$, where $s>0$ and $\mu=\sum_{i=1}^{|\Delta|} 1/i^s$.  It turns out that the exact form of the distribution doesn't matter as much as the fact that it is common among users.

\begin{theorem}
When $p^v = p^w$, for all $v,w\in U$,
\begin{displaymath}
E[\Psi|X_D(u)=d] = b^2+(1-b^2)p^u_d + O(1/n)
\end{displaymath}
\end{theorem}
\begin{proof}
Let $p$ be the common destination distribution.  The expected assigned probability can be expressed as:
\begin{multline} \label{eq:typdist}
E[\Psi|X_D(u)=d] = b^2 + b(1-b)p^u_d +\\
(1-b)\sum_{s=1}^n b^{n-s} (1-b)^{s-1} \sum_{t=0}^{s} (1-b)^{s-t} b^t \binom{n-1}{s-1} \cdot \\
\left[\binom{s-1}{t-1} \sum_{\Delta \in D^t:\Delta_1=d} \prod_{i=2}^t p_{\Delta_i} \Psi_4(s,\Delta) + \binom{s-1}{t} \sum_{\Delta \in D^t} \prod_{i=1}^t p_{\Delta_i} \Psi_4(s,\Delta) \right].
\end{multline}

Here, $s$ represents the size of the set of users with unobserved inputs, $t$ represents the size of the subset of those $s$ users that also have observed outputs, $\Delta$ represents the $t$ observed destinations, and $\Psi_4(s, \Delta)$ is the posterior probability.  In this situation, $\Psi$ is unambiguous given $s$ and $\Delta$.  Let $\Delta_d = |\{x\in \Delta : x=d\}|$.  $\Psi_4$ can be expressed simply as:
\begin{align*}
\Psi_4(s,\Delta) &= \frac{\Delta_d (s-1)^{|\Delta|-1} + p_d (s-1)^{|\Delta|}}{s^{|\Delta|}}\\
&= (\Delta_d + p_d(s-t))/s.
\end{align*}

The sum
\begin{equation*}
\sum_{\Delta \in D^t:\Delta_1=d} \prod_{i=2}^t p_{\Delta_i} \Psi_4(s,\Delta)
\end{equation*}
in Equation~\ref{eq:typdist} calculates the expectation for $\Psi_4$ conditioned on $s$ and $t$.  The expression for $\Psi_4$ shows that this expectation depends linearly on the expected value of $\Delta_d$.  $\Delta_d$'s expectation is simply $1+p_d(t-1)$, because one destination in this case is always $d$, and each of the other $t-1$ is $d$ with probability $p_d$.  The sum
\begin{equation*}
\sum_{\Delta \in D^t} \prod_{i=1}^t p_{\Delta_i} \Psi_4(s,\Delta)
\end{equation*}
in Equation~\ref{eq:typdist} similarly depends linearly on the expectation of $\Delta_d$, which in this case is $p_d t$.

With these observations, it is a straightforward calculation to show that the sum over $t$ in Equation~\ref{eq:typdist} is simply
\begin{align*}
b\frac{p_d (s-1)+1}{s} + (1-b) p_d.
\end{align*}

We insert this into Equation~\ref{eq:typdist} and simplify:
\begin{align*}
E[\Psi|X_D(u)=d] = &b^2 + b(1-b)p^u_d + \\
&(1-b)\sum_{s=1}^n b^{n-s} (1-b)^{s-1} \binom{n-1}{s-1} \left[ b\frac{p_d (s-1)+1}{s} + (1-b) p_d \right]\\
=  &b^2 + b(1-b)p^u_d + \\
&(1-b)\left[b\left(p_d + \frac{(1-p_d)(1-(1-b)^{n+1})}{b(n+1)}\right) + (1-b)p_d \right]\\
=  &b^2 + (1-b^2)p^u_d + O(1/n).
\end{align*}
\end{proof}

Our results show that the expected value of the anonymity metric is close to $b^2 + (1-b^2)p^u_d$ for large populations.  This amount matches the lower bound shown in Thm.~\ref{thm:lwrbnd}.

\section{Conclusions and Future Work}
\label{conclusions}
We expect each user of an anonymity network to have a pattern of use.  In order to make guarantees to the user about his anonymity, we need to take this into account when modeling and analyzing the system, especially in light of previous research that indicates that an adversary can learn these usage patterns given enough time.

We perform such an analysis on onion routing.  Onion routing is a successful design used, in the form of the Tor system, by hundreds of thousands of people to protect their security and privacy.  But, because it was designed to be practical and because theory in this area is still relatively young, the formal analysis of its privacy properties has been limited.

We perform our analysis using a simple black-box model in the UC framework. We justify this model by showing that it information-theoretically provides the same anonymity as the onion routing protocol formalized by \citeN{FC07}. Furthermore, it should lend itself to the analysis of other anonymity protocols expressed within the UC framework.  We investigate the relationship anonymity of users and their destinations in this model and measure it using the probability that the adversary assigns to the correct destination of a given user after observing the network.

Our anonymity analysis first shows that a simple, standard approximation to the expected value of the anonymity metric provides a lower bound on it.  Then we consider the worst-case set of user behaviors to give an upper bound on the expected value.  We show that, in the limit as the number of users grows, a user's anonymity is worst either when all other users choose destinations he is unlikely to visit, because that user becomes unique and identifiable, or when that user chooses a destination that all other users prefer, because the adversary mistakes the group's choices for the user's choice.  This worst-case anonymity with an adversary that controls a fraction $b$ of the routers is comparable to the best-case anonymity against an adversary that controls a fraction $\sqrt{b}$.

The worst case is unlikely to be the case for any users; so we investigate anonymity under a more reasonable model of user behavior suggested in the literature.  In it, users select destinations from a common Zipfian distribution.  Our results show that, in this case and in any case with a common distribution, the expected anonymity tends to the best possible, \emph{i.e.} the adversary doesn't usually gain that much knowledge from the other users' actions.

Future work includes extending this analysis to other types of anonymity (such as sender anonymity), extending it to other anonymity networks, and learning more about the belief distribution of the adversary than just its mean.  A big piece of the attack we describe is in learning the users' destination distribution, about which only a small amount of research, usually on simple models, has been done.  The speed with which an adversary can perform this stage of the attack is crucial in determining the validity of our attack model and results.

In response to analyses such as that of {\O}verlier and
Syverson~\shortcite{hs-attack06}, the current Tor design includes
entry guards by default for all circuits. Roughly, this means that,
since about January 2006, each Tor client selects its first onion
router from a small set of nodes that it randomly selects at
initialization. The rationale is that communication patterns of
individuals are what need to be protected. If an entry guard is
compromised, then the percentage of compromised circuits from that
user is much higher. But, without entry guards, it appears that whom
that user communicates with and even at what rate can be fairly
quickly learned by an adversary owning a modest percentage of the Tor
nodes anyway. If no entry guard is compromised, then no circuits from
that user will ever be linked to him. However, if a user expects to be
targeted by a network adversary that can control nodes, he can expect
his entry guards ultimately to be attacked and possibly
compromised. If the destinations he chooses that are most sensitive
are rarely contacted, he may thus be better off choosing first nodes
at random. How can we know which is better? Extending our analysis to
include entry guards will allow us to answer or at least further
illuminate this question.

Our model also assumes that client connections to the network are such
that the initial onion router in a circuit can tell that it is initial
for that circuit. This is true for the overwhelming majority of
traffic on the Tor network today, because most users run clients that
are not also onion routers. However, for circuits that are initiated
at a node that runs an onion router, a first node cannot easily tell whether it is the first node or the second---without resorting
to other attacks of unknown efficacy, \emph{e.g.}, monitoring latency of
traffic moving in each direction in response to traffic moving in the
other direction. Thus, that initiating edge of the black box is
essentially fuzzy. Indeed, this was originally the only intended
configuration of onion routing for this reason \cite{onion-routing:ih96}.  The
addition of clients that do not also function as routers was a later
innovation that was added to increase usability and flexibility
\cite{onion-routing:jsac98,onion-discex00}.  Similarly, peer-to-peer designs such as Crowds
\cite{crowds:tissec} and Tarzan \cite{tarzan:ccs02} derive their security even more
strongly from the inability of the first node to know whether it is
first or not.  Thus, extending our model and analysis to this case
will make it still more broadly applicable.

\appendix
\section{Appendix}
Let $\tilde{f}$ be as defined in Lemma~\ref{ef}.
\begin{lemma} \label{lem:d2f}
$D^2_{s_{d_i}} \tilde{f} \ge 0$.
\end{lemma}
\begin{proof}
Let $i=s_{d_i}$ and $\mu=\nu-m$ for simplicity.  Then
\begin{equation*}
\tilde{f} = \frac{(\nu +i \mu ) (\nu +(\nu -i) \mu )}{p^u_{d_j} \nu (\nu +i \mu ) (1-i+\nu )+(1+i) p^u_{d_i} \nu  (\nu -i \mu +\nu  \mu )+\beta  (\nu +i \mu ) (\nu -i \mu +\nu  \mu )}.
\end{equation*}
The second derivative of $\tilde{f}$ can be expressed as
\begin{equation*}
D^2_{s_{d_i}} \tilde{f} = \frac{N}{D},
\end{equation*}
where
\begin{align*}
N = &-\bigg( (2 (i+j) (-i-j+\mu ) \\
&\qquad \left(-i^3 (p^u_{d_i}-p^u_{d_j}) \mu ^3 ((i+j) (p^u_{d_i}+p^u_{d_j})+\beta  \mu )+ \right. \\
&\qquad 3 i^2 (i+j) \mu^2 (p^u_{d_i}+p^u_{d_j}+p^u_{d_i} \mu ) ((i+j) (p^u_{d_i}+p^u_{d_j})+\beta  \mu )-\\
&\qquad 3 i (i+j)^2 \mu  ((i+j) (p^u_{d_i}+p^u_{d_j})+\beta  \mu ) \left(-p^u_{d_j}+p^u_{d_i} (1+\mu )^2\right)+\\
&\qquad (i+j)^3 \left((i+j) (p^u_{d_i})^2 (1+\mu )^3+p^u_{d_j} ((i+j) p^u_{d_j}+\beta  \mu )+ \right. \\
&\qquad\qquad \left. \left. p^u_{d_i} \left(\beta  \mu  (1+\mu )^3+p^u_{d_j} (2+\mu ) \left(-i-j+2 \mu +(1+i+j) \mu^2\right)\right)\right)\right)\bigg)
\end{align*}
and
\begin{multline*}
D = \bigg((i+j)^2 (p^u_{d_i}+i p^u_{d_i}+p^u_{d_j}+j p^u_{d_j}+\beta )+\\
(i+j) (i (p^u_{d_j}+\beta )+j (p^u_{d_i}+i p^u_{d_i}+i p^u_{d_j}+\beta )) \mu +i j \beta  \mu^2\bigg)^3,
\end{multline*}
substituting $(i+j)$ for $\nu$.  $D$ is clearly positive.  Therefore we must just show that $N$ is non-negative.

We collect terms in $N$ by the coefficients $p_{d_i}$, $p_{d_j}$, and $\beta$:
\begin{align*}
&2 p_{d_j} \beta (i+j)(i+j-\mu ) \mu  (i+j+i \mu )^3 + \\
&2 p_{d_i} \beta (i+j) (i+j-\mu ) \mu  (i+j+j \mu )^3 +\\
&2 (p^u_{d_j})^2 (i+j)^2 (i+j-\mu )(i+j+i \mu )^3 +\\
&2 (p^u_{d_i})^2 (i+j)^2 (i+j-\mu )(i+j+j \mu )^3 +\\
&2 p^u_{d_i} p^u_{d_j} (i+j)^2 (i+j-\mu ) (i+j)(2+\mu )\cdot \\
&\qquad \left(i^2 \left(-1+\mu ^2\right)+j \left(\mu  (2+\mu )+j \left(-1+\mu ^2\right)\right)+i \left(\mu  (2+\mu )-j \left(2+\mu ^2\right)\right)\right).
\end{align*}
The coefficients of the terms in $p_{d_i}$ and $p_{d_j}$ are clearly positive because $i+j = \nu \ge \nu-m = \mu$.

If we collect the remaining terms by $i$ and $j$, we get
\begin{align*}
&i^3 \left((p^u_{d_i})^2+(p^u_{d_j})^2 (1+\mu )^3+p^u_{d_i} p^u_{d_j} \left(-2-\mu +2 \mu ^2+\mu ^3\right)\right) +\\
&j^3 \left((p^u_{d_j})^2+(p^u_{d_i})^2 (1+\mu )^3+p^u_{d_i} p^u_{d_j} \left(-2-\mu +2 \mu ^2+\mu ^3\right)\right) +\\
&i^2 p^u_{d_i} p^u_{d_j} \mu  (2+\mu )^2 +\\
&j^2 p^u_{d_i} p^u_{d_j} \mu  (2+\mu )^2 +\\
&2 i j p^u_{d_i} p^u_{d_j} \mu  (2+\mu )^2 +\\
&3 i^2 j \left((p^u_{d_i})^2 (1+\mu )+(p^u_{d_j})^2 (1+\mu )^2-p^u_{d_i} p^u_{d_j} (2+\mu )\right) +\\
&3 i j^2 \left((p^u_{d_j})^2 (1+\mu )+(p^u_{d_i})^2 (1+\mu )^2-p^u_{d_i} p^u_{d_j} (2+\mu )\right).
\end{align*}
The coefficients for the $i^3$ and $j^3$ terms are clearly non-negative when $\mu \ge 1$.  When $\mu=0$, observe that the coefficients become $(p^u_{d_i}-p^u_{d_j})^2\ge 0$.  The coefficients for the $i^2$, $j^2$, and $ij$ terms are also clearly non-negative.

To show that the $i^2 j$ term is non-negative, we use the fact that $p_{d_i}$ and $p_{d_j}$ are probabilities that sum to at most one.  Let $p_{d_j} = \zeta-p_{d_i}$, $0\le \zeta \le 1$.  Then the coefficient of $i^2 j$ becomes a quadratic function of $p_{d_i}$ with positive second derivative.  Its minimum is at
\begin{equation*}
p_{d_i} = \frac{4 \zeta +5 \zeta  \mu +2 \zeta  \mu ^2}{2 (2+\mu )^2}.
\end{equation*}
The coefficient evaluated at this point is
\begin{equation*}
\frac{\zeta ^2 \mu  \left(8+11 \mu +4 \mu ^2\right)}{4 (2+\mu )^2},
\end{equation*}
which is non-negative.  Therefore, the whole $i^2 j$ term is non-negative.

Similarly, for the $i j^2$ term, we look at its coefficient as a function of $p_{d_i}$ with $p_{d_j} = \zeta-p_{d_i}$.  It is also a quadratic function with positive second derivative.  Its minimum is found at
\begin{equation*}
\frac{4 \zeta +3 \zeta  \mu }{2 (2+\mu )^2}.
\end{equation*}
The coefficient evaluated at this point is
\begin{equation*}
\frac{\zeta ^2 \mu  (8+\mu  (11+4 \mu ))}{4 (2+\mu )^2},
\end{equation*}
which is non-negative.  Therefore, the whole $i j^2$ term is non-negative.  This implies that $N$ is non-negative, and thus that $D^2_{s_{d_i}} \tilde{f}$ is non-negative.
\end{proof}

\bibliographystyle{acmsmall}
\bibliography{tor_prob_tissec}

\begin{thebibliography}{}

\bibitem[\protect\citeauthoryear{Bauer, McCoy, Grunwald, Kohno, and
  Sicker}{Bauer et~al\mbox{.}}{2007}]{bauer:wpes2007}
{\sc Bauer, K.}, {\sc McCoy, D.}, {\sc Grunwald, D.}, {\sc Kohno, T.}, {\sc
  and} {\sc Sicker, D.} 2007.
\newblock Low-resource routing attacks against {Tor}.
\newblock In {\em {Proceedings of the Workshop on Privacy in the Electronic
  Society (WPES 2007)}}. Washington, DC, USA.

\bibitem[\protect\citeauthoryear{Beimel and Dolev}{Beimel and
  Dolev}{2003}]{buses03}
{\sc Beimel, A.} {\sc and} {\sc Dolev, S.} 2003.
\newblock Buses for anonymous message delivery.
\newblock {\em Journal of Cryptology\/}~{\em 16,\/}~1, 25--39.

\bibitem[\protect\citeauthoryear{Brown}{Brown}{2002}]{cebolla}
{\sc Brown, Z.} 2002.
\newblock Cebolla: Pragmatic {IP} anonymity.
\newblock In {\em Proceedings of the 2002 Ottawa Linux Symposium}.

\bibitem[\protect\citeauthoryear{Camenisch and Lysyanskaya}{Camenisch and
  Lysyanskaya}{2005}]{camlys05}
{\sc Camenisch, J.} {\sc and} {\sc Lysyanskaya, A.} 2005.
\newblock A formal treatment of onion routing.
\newblock In {\em Proceedings of {CRYPTO} 2005}. 169--187.

\bibitem[\protect\citeauthoryear{Canetti}{Canetti}{2000}]{cryptoeprint:2000:067}
{\sc Canetti, R.} 2000.
\newblock Universally composable security: A new paradigm for cryptographic
  protocols.
\newblock Cryptology ePrint Archive, Report 2000/067.
\newblock http://eprint.iacr.org/.

\bibitem[\protect\citeauthoryear{Chaum}{Chaum}{1981}]{chaum-mix}
{\sc Chaum, D.} 1981.
\newblock Untraceable electronic mail, return addresses, and digital
  pseudonyms.
\newblock {\em Communications of the ACM\/}~{\em 4,\/}~2, 84--88.

\bibitem[\protect\citeauthoryear{Chaum}{Chaum}{1988}]{CHAUM}
{\sc Chaum, D.} 1988.
\newblock The dining cryptographers problem: Unconditional sender and recipient
  untraceability.
\newblock {\em Journal of Cryptology: The Journal of the International
  Association for Cryptologic Research\/}~{\em 1,\/}~1, 65--75.

\bibitem[\protect\citeauthoryear{Corrigan-Gibbs and Ford}{Corrigan-Gibbs and
  Ford}{2010}]{Corrigan-Gibbs:2010:DAA:1866307.1866346}
{\sc Corrigan-Gibbs, H.} {\sc and} {\sc Ford, B.} 2010.
\newblock Dissent: accountable anonymous group messaging.
\newblock In {\em Proceedings of the 17th ACM conference on Computer and
  communications security (CCS 2010)}. CCS '10. 340--350.

\bibitem[\protect\citeauthoryear{Danezis}{Danezis}{2003}]{statistical-disclosure}
{\sc Danezis, G.} 2003.
\newblock Statistical disclosure attacks: Traffic confirmation in open
  environments.
\newblock In {\em Proceedings of Security and Privacy in the Age of Uncertainty
  ({SEC 2003})}. 421--426.

\bibitem[\protect\citeauthoryear{Danezis and Serjantov}{Danezis and
  Serjantov}{2004}]{DanSer04}
{\sc Danezis, G.} {\sc and} {\sc Serjantov, A.} 2004.
\newblock Statistical disclosure or intersection attacks on anonymity systems.
\newblock In {\em Proceedings of 6th Information Hiding Workshop (IH 2004)}.
  293--308.

\bibitem[\protect\citeauthoryear{D\'{\i}az, Seys, Claessens, and
  Preneel}{D\'{\i}az et~al\mbox{.}}{2002}]{Diaz02}
{\sc D\'{\i}az, C.}, {\sc Seys, S.}, {\sc Claessens, J.}, {\sc and} {\sc
  Preneel, B.} 2002.
\newblock Towards measuring anonymity.
\newblock In {\em Proceedings of the 2nd Privacy Enhancing Technologies
  Workshop (PET 2002)}. 54--68.

\bibitem[\protect\citeauthoryear{Dingledine, Mathewson, and
  Syverson}{Dingledine et~al\mbox{.}}{2004}]{tor-design}
{\sc Dingledine, R.}, {\sc Mathewson, N.}, {\sc and} {\sc Syverson, P.} 2004.
\newblock Tor: The second-generation onion router.
\newblock In {\em Proceedings of the 13th {USENIX} Security Symposium}. USENIX
  Association, 303--319.

\bibitem[\protect\citeauthoryear{Feigenbaum, Johnson, and Syverson}{Feigenbaum
  et~al\mbox{.}}{2007}]{FC07}
{\sc Feigenbaum, J.}, {\sc Johnson, A.}, {\sc and} {\sc Syverson, P.} 2007.
\newblock A model of onion routing with provable anonymity.
\newblock In {\em Proceedings of the 11th Financial Cryptography and Data
  Security Conference (FC 2007)}. 57--71.

\bibitem[\protect\citeauthoryear{Freedman and Morris}{Freedman and
  Morris}{2002}]{tarzan:ccs02}
{\sc Freedman, M.~J.} {\sc and} {\sc Morris, R.} 2002.
\newblock Tarzan: A peer-to-peer anonymizing network layer.
\newblock In {\em Proceedings of the 9th ACM Conference on Computer and
  Communications Security (CCS 2002)}. 193--206.

\bibitem[\protect\citeauthoryear{Goldberg and Shostack}{Goldberg and
  Shostack}{1999}]{freedom1-security}
{\sc Goldberg, I.} {\sc and} {\sc Shostack, A.} 1999.
\newblock Freedom 1.0 security issues and analysis.
\newblock White paper, Zero Knowledge Systems, {Inc.} November.

\bibitem[\protect\citeauthoryear{Goldberg and Shostack}{Goldberg and
  Shostack}{2001}]{freedom1-arch}
{\sc Goldberg, I.} {\sc and} {\sc Shostack, A.} 2001.
\newblock Freedom network 1.0 architecture and protocols.
\newblock White paper, Zero Knowledge Systems, {Inc.} October.
\newblock The attributed date is that printed at the head of the paper. The
  cited work is, however, superceded by documents that came before Oct.\ 2001.
  The appendix indicates a change history with changes last made November 29,
  1999. Also, in \cite{freedom1-security} the same authors refer to a paper
  with a similar title as an ``April 1999 whitepaper''.

\bibitem[\protect\citeauthoryear{Goldschlag, Reed, and Syverson}{Goldschlag
  et~al\mbox{.}}{1996}]{onion-routing:ih96}
{\sc Goldschlag, D.~M.}, {\sc Reed, M.~G.}, {\sc and} {\sc Syverson, P.~F.}
  1996.
\newblock Hiding routing information.
\newblock In {\em Information Hiding: First International Workshop}. 137--150.

\bibitem[\protect\citeauthoryear{Halpern and O'Neill}{Halpern and
  O'Neill}{2005}]{halpern-oneill-2003}
{\sc Halpern, J.~Y.} {\sc and} {\sc O'Neill, K.~R.} 2005.
\newblock Anonymity and information hiding in multiagent systems.
\newblock {\em Journal of Computer Security\/}~{\em 13,\/}~3, 483--514.

\bibitem[\protect\citeauthoryear{Herrmann, Wendolsky, and Federrath}{Herrmann
  et~al\mbox{.}}{2009}]{ccsw09-fingerprinting}
{\sc Herrmann, D.}, {\sc Wendolsky, R.}, {\sc and} {\sc Federrath, H.} 2009.
\newblock Website fingerprinting: attacking popular privacy enhancing
  technologies with the multinomial na\"\i ve-bayes classifier.
\newblock In {\em Proceedings of the 2009 ACM workshop on Cloud computing
  security (CCSW '09)}. 31--42.

\bibitem[\protect\citeauthoryear{Hopper, Vasserman, and Chan-Tin}{Hopper
  et~al\mbox{.}}{2010}]{tissec-latency-leak}
{\sc Hopper, N.}, {\sc Vasserman, E.~Y.}, {\sc and} {\sc Chan-Tin, E.} 2010.
\newblock How much anonymity does network latency leak?
\newblock {\em ACM Transactions on Information and System Security\/}~{\em
  13,\/}~2.

\bibitem[\protect\citeauthoryear{Hughes and Shmatikov}{Hughes and
  Shmatikov}{2004}]{modular-approach}
{\sc Hughes, D.} {\sc and} {\sc Shmatikov, V.} 2004.
\newblock Information hiding, anonymity and privacy: A modular approach.
\newblock {\em Journal of Computer Security\/}~{\em 12,\/}~1, 3--36.

\bibitem[\protect\citeauthoryear{Kate, Zaverucha, and Goldberg}{Kate
  et~al\mbox{.}}{2007}]{pairing:pet2007}
{\sc Kate, A.}, {\sc Zaverucha, G.}, {\sc and} {\sc Goldberg, I.} 2007.
\newblock Pairing-based onion routing.
\newblock In {\em Privacy Enhancing Technologies: 7th International Symposium,
  ({PET} 2007)}. 95--112.

\bibitem[\protect\citeauthoryear{Kesdogan, Agrawal, and Penz}{Kesdogan
  et~al\mbox{.}}{2002}]{limits-open}
{\sc Kesdogan, D.}, {\sc Agrawal, D.}, {\sc and} {\sc Penz, S.} 2002.
\newblock Limits of anonymity in open environments.
\newblock In {\em Proceedings of the 5th Information Hiding Workshop (IH
  2002)}. 53--69.

\bibitem[\protect\citeauthoryear{Kesdogan, Egner, and B\"uschkes}{Kesdogan
  et~al\mbox{.}}{1998}]{stop-and-go}
{\sc Kesdogan, D.}, {\sc Egner, J.}, {\sc and} {\sc B\"uschkes, R.} 1998.
\newblock Stop-and-go {MIX}es: Providing probabilistic anonymity in an open
  system.
\newblock In {\em Proceedings of the 2nd Information Hiding Workshop (IH
  1998)}. 83--98.

\bibitem[\protect\citeauthoryear{Lincoln, Porras, and Shmatikov}{Lincoln
  et~al\mbox{.}}{2004}]{LINCOLN04}
{\sc Lincoln, P.}, {\sc Porras, P.}, {\sc and} {\sc Shmatikov, V.} 2004.
\newblock Privacy-preserving sharing and correlation of security alerts.
\newblock In {\em Proceedings of the 13th USENIX Security Symposium}. 239--254.

\bibitem[\protect\citeauthoryear{{Loesing et al.}}{{Loesing et
  al.}}{2011}]{tor-metrics}
{\sc {Loesing et al.}} 2011.
\newblock Tor metrics portal.
\newblock https://metrics.torproject.org/.

\bibitem[\protect\citeauthoryear{Lynch}{Lynch}{1996}]{LYNCH}
{\sc Lynch, N.~A.} 1996.
\newblock {\em Distributed Algorithms}.
\newblock Morgan Kaufmann Publishers Inc.

\bibitem[\protect\citeauthoryear{Mathewson and Dingledine}{Mathewson and
  Dingledine}{2004}]{e2e-traffic}
{\sc Mathewson, N.} {\sc and} {\sc Dingledine, R.} 2004.
\newblock Practical traffic analysis: Extending and resisting statistical
  disclosure.
\newblock In {\em Proceedings of the 4th Privacy Enhancing Technologies
  workshop (PET 2004)}. 17--34.

\bibitem[\protect\citeauthoryear{Mauw, Verschuren, and de~Vink}{Mauw
  et~al\mbox{.}}{2004}]{MAUW}
{\sc Mauw, S.}, {\sc Verschuren, J.}, {\sc and} {\sc de~Vink, E.} 2004.
\newblock A formalization of anonymity and onion routing.
\newblock In {\em Proceedings of the 9th European Symposium on Research in
  Computer Security (ESORICS 2004)}. 109--124.

\bibitem[\protect\citeauthoryear{McLachlan, Tran, and Hopper}{McLachlan
  et~al\mbox{.}}{2009}]{ccs09-torsk}
{\sc McLachlan, J.}, {\sc Tran, A.}, {\sc and} {\sc Hopper, N.} 2009.
\newblock Scalable onion routing with {Torsk}.
\newblock In {\em Proceedings of the 16th ACM Conference on Computer and
  Communications Security (CCS'09)}. ACM Press, 590--599.

\bibitem[\protect\citeauthoryear{Mittal and Borisov}{Mittal and
  Borisov}{2009}]{ccs09-shadowwalker}
{\sc Mittal, P.} {\sc and} {\sc Borisov, N.} 2009.
\newblock {ShadowWalker}: Peer-to-peer anonymous communication using redundant
  structured topologies.
\newblock In {\em Proceedings of the 16th ACM Conference on Computer and
  Communications Security (CCS'09)}. 161--172.

\bibitem[\protect\citeauthoryear{Murdoch}{Murdoch}{2006}]{ccs06-hotclockskew}
{\sc Murdoch, S.~J.} 2006.
\newblock Hot or not: Revealing hidden services by their clock skew.
\newblock In {\em Proceedings of the 13th ACM Conference on Computer and
  Communications Security (CCS'06)}. 27--36.

\bibitem[\protect\citeauthoryear{Murdoch and Danezis}{Murdoch and
  Danezis}{2005}]{torta05}
{\sc Murdoch, S.~J.} {\sc and} {\sc Danezis, G.} 2005.
\newblock Low-cost traffic analysis of {Tor}.
\newblock In {\em Proceedings of the 2005 IEEE Symposium on Security and
  Privacy (S\& P 2005)}. 183--195.

\bibitem[\protect\citeauthoryear{Nambiar and Wright}{Nambiar and
  Wright}{2006}]{Salsa}
{\sc Nambiar, A.} {\sc and} {\sc Wright, M.} 2006.
\newblock Salsa: A structured approach to large-scale anonymity.
\newblock In {\em Proceedings of the 13th ACM Conference on Computer and
  Communications Security (CCS 2006)}.

\bibitem[\protect\citeauthoryear{{\O}verlier and Syverson}{{\O}verlier and
  Syverson}{2006}]{hs-attack06}
{\sc {\O}verlier, L.} {\sc and} {\sc Syverson, P.} 2006.
\newblock Locating hidden servers.
\newblock In {\em Proceedings of 2006 {IEEE} Symposium on Security and Privacy
  ({S\& P} 2006)}. IEEE CS, 100--114.

\bibitem[\protect\citeauthoryear{{\O}verlier and Syverson}{{\O}verlier and
  Syverson}{2007}]{overlier:pet2007}
{\sc {\O}verlier, L.} {\sc and} {\sc Syverson, P.} 2007.
\newblock Improving efficiency and simplicty of {Tor} circuit establishment and
  hidden services.
\newblock In {\em Privacy Enhancing Technologies: 7th International Symposium
  ({PET} 2007)}. 134--152.

\bibitem[\protect\citeauthoryear{Pfitzmann and Hansen}{Pfitzmann and
  Hansen}{2000}]{terminology}
{\sc Pfitzmann, A.} {\sc and} {\sc Hansen, M.} 2000.
\newblock Anonymity, unobservability, and pseudonymity: A consolidated proposal
  for terminology.
\newblock Draft.

\bibitem[\protect\citeauthoryear{Reed, Syverson, and Goldschlag}{Reed
  et~al\mbox{.}}{1998}]{onion-routing:jsac98}
{\sc Reed, M.~G.}, {\sc Syverson, P.~F.}, {\sc and} {\sc Goldschlag, D.~M.}
  1998.
\newblock Anonymous connections and onion routing.
\newblock {\em {IEEE} Journal on Selected Areas in Communications\/}~{\em
  16,\/}~4, 482--494.

\bibitem[\protect\citeauthoryear{Reiter and Rubin}{Reiter and
  Rubin}{1998}]{crowds:tissec}
{\sc Reiter, M.} {\sc and} {\sc Rubin, A.} 1998.
\newblock Crowds: Anonymity for web transactions.
\newblock {\em ACM Transactions on Information and System Security
  (TISSEC)\/}~{\em 1,\/}~1, 66--92.

\bibitem[\protect\citeauthoryear{Schneider and Sidiropoulos}{Schneider and
  Sidiropoulos}{1996}]{schneider96}
{\sc Schneider, S.} {\sc and} {\sc Sidiropoulos, A.} 1996.
\newblock {CSP} and anonymity.
\newblock In {\em Proceedings of the 1st European Symposium on Research in
  Computer Security (ESORICS 1996)}. 198--218.

\bibitem[\protect\citeauthoryear{Serjantov and Danezis}{Serjantov and
  Danezis}{2002}]{Serj02}
{\sc Serjantov, A.} {\sc and} {\sc Danezis, G.} 2002.
\newblock Towards an information theoretic metric for anonymity.
\newblock In {\em Proceedings of the 2nd Privacy Enhancing Technologies
  Workshop (PET 2002)}. 41--53.

\bibitem[\protect\citeauthoryear{Shmatikov}{Shmatikov}{2004}]{SHMAT}
{\sc Shmatikov, V.} 2004.
\newblock Probabilistic model checking of an anonymity system.
\newblock {\em Journal of Computer Security\/}~{\em 12,\/}~3-4, 355--377.

\bibitem[\protect\citeauthoryear{Shmatikov and Wang}{Shmatikov and
  Wang}{2006}]{ShWa-Relationship}
{\sc Shmatikov, V.} {\sc and} {\sc Wang, M.-H.} 2006.
\newblock Measuring relationship anonymity in mix networks.
\newblock In {\em {Proceedings of the 5th ACM Workshop on Privacy in the
  Electronic Society (WPES 2006)}}. 59--62.

\bibitem[\protect\citeauthoryear{Syverson, Reed, and Goldschlag}{Syverson
  et~al\mbox{.}}{2000}]{onion-discex00}
{\sc Syverson, P.}, {\sc Reed, M.}, {\sc and} {\sc Goldschlag, D.} 2000.
\newblock Onion routing access configurations.
\newblock In {\em Proceedings of the DARPA Information Survivability Conference
  and Exposition (DISCEX 2000)}. 34--40.

\bibitem[\protect\citeauthoryear{Syverson, Tsudik, Reed, and Landwehr}{Syverson
  et~al\mbox{.}}{2000}]{onion-routing:pet2000}
{\sc Syverson, P.}, {\sc Tsudik, G.}, {\sc Reed, M.}, {\sc and} {\sc Landwehr,
  C.} 2000.
\newblock {Towards an Analysis of Onion Routing Security}.
\newblock In {\em Designing Privacy Enhancing Technologies: Workshop on Design
  Issues in Anonymity and Unobservability}. 96--114.

\bibitem[\protect\citeauthoryear{Syverson and Stubblebine}{Syverson and
  Stubblebine}{1999}]{GROUP}
{\sc Syverson, P.~F.} {\sc and} {\sc Stubblebine, S.~G.} 1999.
\newblock Group principals and the formalization of anonymity.
\newblock In {\em Proceedings of the 1st World Congress on Formal Methods
  (FM'99), Vol. I}. 814--833.

\bibitem[\protect\citeauthoryear{T\'oth, Horn\'ak, and Vajda}{T\'oth
  et~al\mbox{.}}{2004}]{TOTH}
{\sc T\'oth, G.}, {\sc Horn\'ak, Z.}, {\sc and} {\sc Vajda, F.} 2004.
\newblock Measuring anonymity revisited.
\newblock In {\em Proceedings of the 9th Nordic Workshop on Secure IT Systems}.
  85--90.

\bibitem[\protect\citeauthoryear{Wikstr{\"o}m}{Wikstr{\"o}m}{2004}]{Wikstrom04}
{\sc Wikstr{\"o}m, D.} First Theory of Cryptography Conference (TCC 2004).
\newblock A universally composable mix-net.
\newblock In {\em TCC}. 317--335.

\bibitem[\protect\citeauthoryear{Wright, Adler, Levine, and Shields}{Wright
  et~al\mbox{.}}{2004}]{WRIGHT}
{\sc Wright, M.~K.}, {\sc Adler, M.}, {\sc Levine, B.~N.}, {\sc and} {\sc
  Shields, C.} 2004.
\newblock The predecessor attack: An analysis of a threat to anonymous
  communications systems.
\newblock {\em ACM Transactions on Information and Systems Security\/}~{\em
  7,\/}~4, 489--522.

\end{thebibliography}
\end{document}